\documentclass[11pt]{article}

\usepackage{fullpage}
\usepackage{times}
\usepackage{amssymb,amsmath}
\usepackage{paralist}
\usepackage{bm}
\usepackage{epsfig}
\usepackage{amsfonts}
\usepackage[mathscr]{eucal}
\usepackage{xspace}
\usepackage{url}

\newtheorem{theorem}{Theorem}

\newtheorem{definition}{Definition}[section]

\newtheorem{assumption}{Assumption}

\newtheorem{lemma}{Lemma}

\newcommand{\QED}{\ensuremath{\Box}}
\newenvironment{proof}{\paragraph{{\sc Proof:}}}{\hfill \QED \medskip}

\newenvironment{sketch}{\paragraph{\sc Proof sketch:}}{\hfill $\QED$ \medskip}

\newcommand{\superparagraphb}[1]{\noindent{\bf{#1}}}

\newcommand{\from}{\leftarrow}

\newcommand{\pr}[1]{\ensuremath{{\mathrm{Pr}}[#1]}}

\newcommand{\poly}{\ensuremath{\text{\em poly}}}

\newcommand{\NOTE}[1]{}
\newcounter{todos}
\newcommand{\namedref}[3]{#1~\ref{#2:#3}}

\newcommand{\sectionref}[1]{\namedref{Section}{sec}{#1}}

\newcommand{\superparagraph}[1]{\smallskip\noindent {\bf #1}}

\newcommand{\ZZ}{\ensuremath{\mathbb Z}\xspace}

\newcommand{\functionality}[1]{\ensuremath{{\mathcal #1}}\xspace}
\newcommand{\thing}[1]{\ensuremath{{\text{\sf #1}}}\xspace}
\newcommand{\keyword}[1]{\ensuremath{{\text{\sc #1}}}\xspace}
\newcommand{\command}[1]{\ensuremath{{\text{\sf #1}}}\xspace}

\newcommand{\class}[1]{\ensuremath{\text{\bf{#1}}}\xspace}

\newcommand{\NC}[1]{\ensuremath{\class{NC}^{#1}}\xspace}
\newcommand{\NCz}{\ensuremath{\NC0}\xspace}

\newcommand{\codegen}{\ensuremath{\mathcal{G}}\xspace}
\newcommand{\randcodegen}{\ensuremath{\mathcal{G}_{\command{Rand}}}\xspace}
\newcommand{\rscodegen}{\ensuremath{\mathcal{G}_{\command{RS}}}\xspace}
\newcommand{\ringcodegen}{\ensuremath{\mathcal{G}_{\command{Ring}}}\xspace}
\newcommand{\lincode}[4]{\ensuremath{\mathcal{E}_{(#1,#3)}^{{#2}}({#4})}\xspace}
\newcommand{\encode}[4]{\ensuremath{\mathsf{Encode}_{(#1,#3)}^{{#2}}({#4})}\xspace}

\newcommand{\add}{\command{add}}
\newcommand{\subtract}{\command{subtract}}
\newcommand{\mult}{\command{multiply}}

\newcommand{\sample}{\command{sample}}
\newcommand{\one}{\command{one}}
\newcommand{\operation}{\command{op}}
\newcommand{\invert}{\command{invert}}

\newcommand{\ringid}{\command{id}}
\newcommand{\rep}{\command{label}}

\newcommand{\ideal}{\keyword{ideal}}
\newcommand{\real}{\keyword{real}}

\newcommand{\ring}{\ensuremath{R}\xspace}
\newcommand{\field}{\ensuremath{F}\xspace}

\newcommand{\ringfamily}{\ensuremath{\mathscr{R}}\xspace}
\newcommand{\fieldfamily}{\ensuremath{\mathscr{F}}\xspace}
\newcommand{\ringalgo}{\ensuremath{\mathsf{R}}\xspace}

\newcommand{\hash}{\ensuremath{\mathcal H}\xspace}
\newcommand{\unif}{\ensuremath{\mathcal U}\xspace}
\newcommand{\statdiff}{\ensuremath{\Delta}\xspace}

\newcommand{\statencode}{\ensuremath{\mathsf{Enc}^\ringfamily_n}\xspace}

\newcommand{\scode}{\ensuremath{\mathcal S}\xspace}

\newcommand{\Ffunc}{\functionality{F}}
\newcommand{\func}[1]{\ensuremath{\Ffunc_{#1}}\xspace}
\newcommand{\namedfunc}[1]{\func{\text{\sf #1}}}
\newcommand{\env}{\thing{Env}}
\newcommand{\envx}{\thing{Env}}
\newcommand{\adv}{\thing{Adv}}
\newcommand{\simx}{\thing{Sim}}
\newcommand{\simu}{\ensuremath{\widetilde{\simx}}\xspace}

\newcommand{\ckt}{\ensuremath{C}\xspace}
\newcommand{\prot}{\ensuremath{\pi}\xspace}
\newcommand{\fmult}{\namedfunc{pdt-shr}}
\newcommand{\fmultx}{\ensuremath{{\Ffunc}'_{\text{\sf pdt-shr}}}\xspace}
\newcommand{\pmult}{\ensuremath{\widetilde{\pi}^\OT}\xspace}
\newcommand{\pmultx}{\ensuremath{\pi^\OT}\xspace}
\newcommand{\pmultbasicx}{\ensuremath{\rho^\OT}\xspace}
\newcommand{\pmultlcx}{\ensuremath{\sigma^\OT}\xspace}
\newcommand{\envad}{\ensuremath{\widetilde{\env}}\xspace}

\newcommand{\fckt}[1]{\ensuremath{\namedfunc{\ckt}^{#1}}\xspace}
\newcommand{\fcktid}[1]{\ensuremath{\namedfunc{\ckt,\ringid}^{#1}}\xspace}

\newcommand{\fmultt}{\ensuremath{{\Ffunc}^{\field^t}_{\text{\sf pdt-shr}}}\xspace}
\newcommand{\pmulttx}{\ensuremath{\tau^\OT}\xspace}
\newcommand{\at}{\ensuremath{\mathbf{a}}\xspace}
\newcommand{\bt}{\ensuremath{\mathbf{b}}\xspace}
\newcommand{\zt}{\ensuremath{{z}}\xspace}
\newcommand{\vn}{\ensuremath{{v}}\xspace}

\newcommand{\alphat}{\ensuremath{{\bm \alpha}}\xspace}
\newcommand{\pa}{\ensuremath{P_a}\xspace}
\newcommand{\pb}{\ensuremath{P_b}\xspace}
\newcommand{\pc}{\ensuremath{P_r}\xspace}
\newcommand{\pz}{\ensuremath{Q}\xspace}

\newcommand{\x}[1]{\ensuremath{\zeta_{#1}}\xspace}
\newcommand{\y}[1]{\ensuremath{\vartheta_{#1}}\xspace}

\newcommand{\OT}{\ensuremath{\text{\sf OT}}\xspace}

\newcommand{\alice}{\ensuremath{A}\xspace}
\newcommand{\bob}{\ensuremath{B}\xspace}

\title{
Secure Arithmetic Computation with No Honest Majority}
\author{
     {Yuval Ishai}
     \thanks{Technion, Israel and University of California, Los Angeles. {\tt
      yuvali@cs.technion.il}}
   \and
     {Manoj Prabhakaran}
    \thanks{University of Illinois, Urbana-Champaign. {\tt mmp@cs.uiuc.edu}}
   \and
     {Amit Sahai}
    \thanks{University of California, Los Angeles. {\tt sahai@cs.ucla.edu}}
 }

\begin{document}
\maketitle

 \begin{abstract}

We study the complexity of securely evaluating arithmetic circuits over finite rings. This question is motivated by natural secure computation tasks. Focusing mainly on the case of {\em two-party} protocols with security against {\em malicious} parties, our main goals are to: (1) only make black-box calls to the ring operations and standard cryptographic primitives, and (2) minimize the number of such black-box calls as well as the communication overhead.

We present several solutions which differ in their efficiency, generality, and underlying intractability assumptions. These include:
\begin{itemize}
  \item An {\em unconditionally secure} protocol in the OT-hybrid model which makes a black-box use of an arbitrary ring $R$, but where the number of ring operations grows linearly with (an upper bound on) $\log|R|$.

  \item Computationally secure protocols in the OT-hybrid model which make a black-box use of an underlying ring, and in which the number of ring operations does not grow with the ring size.
The protocols rely on variants of previous intractability assumptions related to linear codes.
In the most efficient instance of these protocols, applied to a suitable class of fields, the (amortized) communication cost is a constant number of field elements per multiplication gate and the computational cost is dominated by $O(\log k)$ field operations per gate, where $k$ is a security parameter.
 These results extend a previous approach of Naor and Pinkas for secure polynomial evaluation ({\em SIAM J.\ Comput.}, 35(5), 2006).

  \item A protocol for the rings $\mathbb{Z}_m=\mathbb{Z}/m\mathbb{Z}$ which only makes a black-box use of a homomorphic encryption scheme.  When $m$ is prime, the (amortized) number of calls to the encryption scheme for each gate of the circuit is constant.
\end{itemize}
All of our protocols are in fact {\em UC-secure} in the OT-hybrid model and can be generalized to {\em multiparty} computation with an arbitrary number of malicious parties.
\end{abstract}

 \thispagestyle{empty}
 \newpage
 \setcounter{page}{1}

\section{Introduction}
\label{sec:intro}

This paper studies the complexity of secure multiparty computation (MPC) tasks which involve {\em arithmetic} computations. Following the general feasibility results from the 1980s~\cite{Yao86,GoldreichMiWi87,BenOrGoWi88,ChaumCrDa88}, much research in this area shifted to efficiency questions, with a major focus on the efficiency of securely distributing natural computational tasks that arise in the ``real world''. In many of these cases, some inputs, outputs, or intermediate values in the computation are integers, finite-precision reals, matrices, or elements of a big finite ring, and the computation involves arithmetic operations in this ring.  To name just a few examples from the MPC literature, such arithmetic computations are useful in the contexts of distributed generation of cryptographic keys~\cite{BonehF01,FrankelMY98,PoupardS98,Gilboa99,AlgesheimerCS02}, privacy-preserving data-mining and statistics~\cite{LindellP02,CanettiIKRRW01}, comparing and matching data~\cite{NaorP06,FreedmanNP04,HazayL08}, 
auctions and mechanism design~\cite{NaorPS99,DamgardFKNT06,Toft07thesis,BogetoftChDaGeJaKrNiNiNiPaScTo08}, and distributed linear algebra computations~\cite{CramerD01,NissimW06,KiltzMWF07,CramerKP07,MohasselW08}.

This motivates the following question:
\begin{quote}
What is the complexity of securely evaluating a given arithmetic circuit $C$ over a given finite ring $R$?
\end{quote}

Before surveying the state of the art, some clarifications are in place.

\superparagraph{Arithmetic circuits.} An arithmetic circuit over a ring is defined similarly to a standard boolean circuit, except that the inputs and outputs are ring elements rather than bits and gates are labeled by the ring operations  \add, \subtract, and \mult. (Here and in the following, by ``ring'' we will refer to a finite ring by default.)
In the current context of distributed computations, the inputs and outputs of the circuit are annotated with the parties to which they belong. Thus, the circuit $C$ together with the ring $R$ naturally define a multi-party arithmetic functionality $C^R$. Note that arithmetic computations over the integers or finite-precision reals can be embedded into a sufficiently large finite ring or field, provided that there is an a-priori upper bound on the bit-length of the output.
See Section~\ref{sec-discuss} for further discussion of the usefulness of arithmetic circuits and some extensions of this basic model to which our results apply.

\superparagraph{Secure computation model.} The main focus of this paper is on secure {\em two-party} computation or, more generally, MPC with an arbitrary number of malicious parties. (In this setting it is generally impossible to guarantee output delivery or even fairness, and one has to settle for allowing the adversary to abort the protocol after learning the output.) Our protocols are described in the ``OT-hybrid model,'' namely in a model that allows parties to invoke an ideal oblivious transfer (OT) oracle~\cite{Rabin81,EvenGoLe85,Goldreich04book}. This has several advantages in generality and efficiency, see~\cite{IshaiPrSa08} and Section~\ref{sec-discuss} below for discussion.

\superparagraph{Ruling out the obvious.} An obvious approach for securely realizing an arithmetic computation $C^R$ is by first designing an equivalent {\em boolean} circuit $C'$ which computes the same function on a binary representation of the inputs, and then using standard MPC protocols for realizing $C'$. The main disadvantage of such an approach is that it typically becomes very inefficient when $R$ is large.
One way to rule out such an approach, at least given the current state of the art, is to require the communication complexity to grow at most linearly with $\log|R|$. (Note that even in the case of finite fields with $n$-bit elements, the size of the best known boolean multiplication circuits is $\omega(n\log n)$; the situation is significantly worse for other useful rings, such as matrix rings.)

A cleaner way for ruling out such an approach, which is of independent theoretical interest, is by restricting protocols to only make a {\em black-box} access to the ring $R$.
That is, $\Pi$ securely realizes $C$ if $\Pi^R$ securely realizes $C^R$ for every finite ring $R$ and {\em every representation of elements in $R$}.\footnote{When considering computational security we will require representations to be {\em computationally efficient}, in the sense that given identifiers of two ring elements $a,b$ one can efficiently compute the identifiers of $a+b$, $a-b$, and $a\cdot b$.}
This black-box access to $R$ enables $\Pi$ to perform ring operations and sample random ring elements, but the correspondence between ring elements and their identifiers (or even the exact size of the ring) will be unknown to the protocol.\footnote{Note that it is not known how to efficiently learn the structure of a ring using a black box access to ring operations, even in the special case of general finite fields~\cite{BonehL96,MaurerR07}.}
When considering the special case of fields, we allow by default the protocol $\Pi$ to access an inversion oracle.

\subsection{Previous Work}
\label{sec-related}

In the setting of MPC {\em with honest majority}, most protocols from the literature can make a black-box use of an arbitrary {\em field}.
An extension to arbitrary black-box rings was given in~\cite{CramerFIK03}, building on previous black-box secret sharing techniques of~\cite{DesmedtF89,CramerF02}.

In the case of secure two-party computation and MPC with no honest majority, most protocols from the literature apply to boolean circuits. Below we survey some previous approaches from the literature that apply to secure arithmetic computation with no honest majority.

In the semi-honest model, it is easy to employ any homomorphic encryption scheme with plaintext group $\ZZ_m$ for performing arithmetic MPC over $\ZZ_m$. (See, e.g.,~\cite{AbadiF90,CanettiIKRRW01}.) An alternative approach, which relies on oblivious transfer and uses the standard binary representation of elements in $\ZZ_m$, was employed in~\cite{Gilboa99}. These protocols make a black-box use of the underlying cryptographic primitives but do not make a black-box use of the underlying ring. Applying the general compilers of~\cite{GoldreichMiWi87,CanettiLiOsSa02} to these protocols in order to obtain security in the malicious model would result in inefficient protocols which make a non-black-box use of the underlying cryptographic primitives (let alone the ring).

In the {\em malicious model}, protocols for secure arithmetic computation based on {\em threshold} homomorphic encryption were given in~\cite{CramerDN01,DamgardN03}\footnote{While~\cite{CramerDN01,DamgardN03} refer to the case of robust MPC in the presence of an honest majority, these protocols can be easily modified to apply to the case of MPC with no honest majority. We note that while a main goal of these works was to minimize the growth of complexity with the number of parties, we focus on minimizing the complexity in the two-party case.}  (extending a similar protocol for the semi-honest model from~\cite{FranklinH96}).  These protocols provide the most practical general solutions for secure arithmetic two-party computation we are aware of, requiring a constant number of modular exponentiations for each arithmetic gate. On the down side, these protocols require a nontrivial setup of keys which is expensive to distribute. Moreover, similarly to all protocols described so far, they rely on special-purpose zero-knowledge proofs and specific number-theoretic assumptions and thus do not make a black-box use of the underlying cryptographic primitives, let alone a black-box use of the ring.

The only previous approach which makes a black-box use of an underlying ring (as well as a black-box use of OT) was suggested by Naor and Pinkas~\cite{NaorP06} in the context of secure polynomial evaluation. Their protocol can make a black-box use of any {\em field} (assuming an inversion oracle), and its security is related to the conjectured intractability of decoding Reed-Solomon codes with a sufficiently high level of random noise. The protocol from~\cite{NaorP06} can be easily used to obtain general secure protocols for arithmetic circuits in the {\em semi-honest}
model. However, extending it to allow full simulation-based security in the malicious model (while still making only a black-box use of the underlying field) is not straightforward. (Even in the special case of secure polynomial evaluation, an extension to the malicious model suggested in~\cite{NaorP06} only considers {\em privacy} rather than full simulation-based security.)

Finally, we note that Yao's garbled circuit technique~\cite{Yao86}, which is essentially the only known technique for constant-round secure computation of general functionalities, does not have a known arithmetic analogue. Thus, in all general-purpose protocols for secure arithmetic computation (including the ones presented in this work) the round complexity must grow with the multiplicative depth\footnote{The {\em multiplicative depth} of a circuit is the maximal number of multiplication gates on a path from an input to an output.}
of $C$.

\subsection{Our Contribution}
\label{sec-contribution}

We study the complexity of general secure arithmetic computation over finite rings in the presence of an arbitrary number of malicious parties. We are motivated by the following two related goals.
\begin{itemize}
  \item {\em Black-box feasibility}: only make a black-box use of an underlying ring $R$ or field $F$ and standard cryptographic primitives;
  \item {\em Efficiency}: minimize the number of such black-box calls, as well as the communication overhead.
\end{itemize}

For simplicity, we do not attempt to optimize the dependence of the complexity on the number of parties, and restrict the following discussion to the two-party case.

We present several solutions which differ in their efficiency, generality, and underlying intractability assumptions. Below we describe the main protocols along with their efficiency and security features. An overview of the underlying techniques is presented in Section~\ref{sec-techniques}.

\superparagraph{An unconditionally secure protocol.} We present an {\em unconditionally secure} protocol in the OT-hybrid model which makes a {\em black-box} use of an {\em arbitrary} finite ring $R$, but where the number of ring operations and the number of ring elements being communicated grow linearly with (an upper bound on) $\log|R|$.  (We assume for simplicity that an upper bound on $\log|R|$ is given by the ring oracle, though such an upper bound can always be inferred from the length of the strings representing ring elements.)
More concretely, the number of ring operations for each gate of $C$ is $\poly(k)\cdot\log|R|$, where $k$ is a statistical security parameter. This gives a two-party analogue for the MPC protocol over black-box rings from~\cite{CramerFIK03}, which requires an honest majority (but does not require the number of ring operations to grow with $\log |R|$).

\superparagraph{Protocols based on noisy linear encodings.}
Motivated by the goal of reducing the overhead of the previous protocol,
we present a general approach for deriving secure arithmetic computation protocols over a ring $R$ from linear codes over $R$.
The (computational) security of the protocols relies on intractability assumptions related to the hardness of decoding in the presence of random noise. These protocols generalize and extend in several ways the previous approach of Naor and Pinkas for secure polynomial evaluation~\cite{NaorP06} (see Section~\ref{sec-techniques} for discussion). Using this approach, we obtain the following types of protocols in the OT-hybrid model.
\begin{itemize}
  \item A protocol which makes a black-box use of an arbitrary {\em field} $F$, in which the number of field operations (and field elements being communicated) does not grow with the field size. More concretely, the number of field operations for each gate of $C$ is bounded by a fixed polynomial in the security parameter $k$, independently of $|F|$.
The underlying assumption is related to the conjectured intractability of decoding a random linear code\footnote{ The above efficiency feature requires that random linear codes remain hard to decode even over very large fields. Note, however, that $\log |F|$ is effectively restricted by the running time of the adversary, which is (an arbitrarily large) polynomial in $k$. The assumption can be relaxed if one allows the number of ring operation to moderately grow with $\log|F|$.}  over $F$. Our assumption is implied by the assumption that a noisy codeword in a random linear code over $F$ is pseudorandom.
Such a pseudorandomness assumption follows from the average-case hardness of {\em decoding} a random linear code when the field size is polynomial in $k$ (see~\cite{BlumFKL93,ApplebaumIK07} for corresponding reductions in the binary case).
  \item A variant of the previous protocol which makes a black-box use of an arbitrary {\em ring} $R$, and in particular does not rely on inversion. This variant is based on families of linear codes over rings in which decoding in the presence of erasures can be done efficiently, and for which decoding in the presence of (a suitable distribution of) random noise seems intractable.
  \item The most efficient protocol we present relies on the intractability of decoding Reed-Solomon codes with a (small) constant rate in the presence of a (large) constant fraction of noise.\footnote{The precise intractability assumption we use is similar in flavor to an assumption used in~\cite{NaorP06} for evaluating polynomials of degree $d\ge 2$. With a suitable choice of parameters, our assumption is implied by a natural pseudorandomness variant of the assumption from~\cite{NaorP06}, discussed in~\cite{KiayiasY08}. The assumption does not seem to be affected by the recent progress on list-decoding Reed-Solomon codes and their variants~\cite{GuruswamiS99,CoppersmithS03,BleichenbacherKY07,ParvareshVa05}.}
      The amortized communication cost is a constant number of field elements per multiplication gate. (Here and in the following, when we refer to ``amortized'' complexity we ignore an additive term that may depend polynomially on the security parameter and the circuit depth, but not on the circuit size. In most natural instances of large circuits this additive term does not form an efficiency bottleneck.)

      A careful implementation yields protocols whose amortized computational cost is $O(\log k)$ field operations per gate, where $k$ is a security parameter, assuming that the field size is super-polynomial in $k$.
      In contrast, protocols which are based on homomorphic encryption schemes (such as~\cite{CramerDN01} or the ones obtained in this work) apply modular exponentiations, which require $\Omega(k+\log|F|)$ ring multiplications per gate, in a ciphertext ring which is larger than $F$. This is the case even in the semi-honest model. Compared to the ``constant-overhead'' protocol from~\cite{IshaiKOS08} (applied to a boolean circuit realizing $C^F$), our protocol has better communication complexity and relies on a better studied assumption, but its asymptotic computational complexity is worse by an $O(\log k)$ factor when implemented in the {\em boolean} circuit model.
\end{itemize}

\superparagraph{Protocols making a black-box use of homomorphic encryption.}
For the case of rings of the form $\mathbb{Z}_m=\mathbb{Z}/m\mathbb{Z}$ (with the standard representation) we present a protocol which makes a black-box use of any homomorphic encryption scheme with plaintext group $\ZZ_m$. Alternatively, the protocol can make a black-box use of homomorphic encryption schemes in which the plaintext group is determined by the key generation algorithm, such as those of Paillier~\cite{Paillier99} or Damg{\aa}rd-Jurik~\cite{DamgardJ02}. In both variants of the protocol, the (amortized) number of communicated ciphertexts and calls to the encryption scheme for each gate of $C$ is constant, assuming that $m$ is prime.
This efficiency feature is comparable to the protocols from~\cite{CramerDN01,DamgardN03} discussed in Section~\ref{sec-related} above. Our protocols have the advantages of using a more general primitive and only making a {\em black-box} use of this primitive (rather than relying on special-purpose zero-knowledge protocols). Furthermore, the additive term which we ignore in the above ``amortized'' complexity measure seems to be considerably smaller than the cost of distributing the setup of the threshold cryptosystem required by~\cite{CramerDN01}. 

Both variants of the protocol can be naturally extended to the case of matrix rings $\mathbb{Z}^{n\times n}_m$, increasing the communication complexity by a factor of $n^2$. (Note that emulating matrix operations via basic arithmetic operations over $\ZZ_m$ would result in a bigger overhead, corresponding to the complexity of matrix multiplication.) Building on the techniques from~\cite{MohasselW08}, this protocol can be used to obtain efficient protocols for secure linear algebra which make a black-box use of homomorphic encryption and achieve simulation-based security against malicious parties (improving over similar protocols with security against {\em covert} adversaries~\cite{AumannLi07} recently presented in~\cite{MohasselW08}).

\medskip

All of our protocols are in fact {\em UC-secure} in the OT-hybrid model and can be generalized to {\em multiparty} computation with an arbitrary number of malicious parties. The security of the protocols also holds against {\em adaptive} adversaries, assuming that honest parties may erase data. (This is weaker than the standard notion of adaptive security~\cite{CanettiFGN96} which does not rely on data erasure.) The {\em round complexity} of all the protocols is a constant multiple of the multiplicative depth of $C$.

\subsection{Techniques}
\label{sec-techniques}

Our results build on a recent technique from~\cite{IshaiPrSa08} (which was inspired by previous ideas from~\cite{IshaiKuOsSa07} and also \cite{HarnikIsKuNi08}).
The main result of~\cite{IshaiPrSa08} constructs a secure two-party protocol for a functionality $f$ in the OT-hybrid model by making a {\em black-box} use of the following two ingredients: (1) an {\em outer MPC protocol} which realizes $f$ using $k$ additional ``servers'', but only needs to tolerate a {\em constant fraction} of {\em malicious} servers; and (2) an {\em inner two-party protocol} which realizes in the {\em semi-honest OT-hybrid model} a reactive two-party functionality defined (in a black-box way) by the outer protocol. The latter functionality is essentially a distributed version of the algorithm run by the servers in the outer protocol.

Because of the black-box nature of this construction, if both ingredients make a black-box use of $R$ and/or a black-box use of cryptographic primitives, then so does the final two-party protocol.

Given the above, it remains to find good instantiations for the outer and inner protocols. Fortunately, good instances of the outer protocol already exist in the literature. In the case of general black-box rings, we can use the protocol of~\cite{CramerFIK03}. In the case of fields, we can use a variant of the protocol from~\cite{DamgardIs06} for better efficiency. This protocol has an amortized
communication cost of a constant number of field elements for each multiplication gate in the circuit.  In terms of computational overhead, a careful implementation incurs an amortized overhead of $O(\log k)$ field operations per gate, where $k$ is a security parameter, assuming that the field size is superpolynomial in $k$.  (The overhead is dominated by the cost of Reed-Solomon encoding over the field.)

Our final protocols are obtained by combining the above outer protocols with suitable implementations of the inner protocol. Our main technical contribution is in suggesting concrete inner protocols which yield the required security and efficiency features.

Similarly to~\cite{IshaiPrSa08}, the inner protocols corresponding to the outer protocols we employ require to securely compute, in the semi-honest model, multiple instances
of a simple ``product-sharing'' functionality, in which Alice holds a ring element $a$, Bob holds a ring element $b$, and the output is an additive secret sharing of $ab$. (The efficient version of the outer protocol requires the inner protocol to perform only a constant amortized number of product-sharings per multiplication gate. All other computations, including ones needed for handling addition gates, are done locally and do not require interaction.) In~\cite{IshaiPrSa08} such a product-sharing protocol is implemented by applying the GMW protocol~\cite{GoldreichMiWi87} (in the semi-honest OT-hybrid model) to the binary representation of the inputs. This does not meet our current feasibility and efficiency goals.

Below we sketch the main ideas behind different product-sharing protocols on which we rely, which correspond to the main protocols described in Section~\ref{sec-contribution}.

\superparagraph{Unconditionally secure product-sharing.}  In our unconditionally secure protocol, Bob breaks his input $b$ into $n$ additive shares and uses them to generate $n$ pairs of ring elements, where in each pair one element is a share of $b$ and the other is a random ring element. (The location of the share of $b$ in each pair is picked at random and is kept secret by Alice. Note that additive secret-sharing can be done using a black-box access to the ring oracle.) Bob  sends these $n$ pairs to Alice. Alice multiplies each of the $2n$ ring elements (from the left) by her input $a$, and subtracts from each element in the $i$-th pair a random ring element $t_i$. This results in $n$ new pairs. Bob retrieves from each pair the element corresponding to the original additive share of $b$ by using $n$ invocations of the OT oracle. Bob outputs the sum of the $n$ ring elements she obtained, and Alice outputs $\sum_{i=1}^n t_i$.

It is easy to verify that the protocol has the correct output distribution.  The security of the protocol can be analyzed using the Leftover Hash Lemma~\cite{ImpagliazzoLL89}. (Similar uses of this lemma were previously made
in~\cite{ImpagliazzoNaor96,IshaiKOS06}.) Specifically, the protocol is statistically secure when $n>\log_2|R|+k$. We note that in light of efficient algorithms for low-density instances of subset sum~\cite{LagariasO85}, one cannot hope to obtain significant efficiency improvements by choosing a smaller value of $n$ and settling for computational security.

\superparagraph{Product-sharing from linear codes.} Our construction for black-box fields generalizes the previous approach of Naor and Pinkas~\cite{NaorP06} in a natural way. The high level idea is as follows. Bob sends to Alice a {\em noisy} randomized linear encoding (or noisy linear secret-sharing) of $b$ which is assumed to hide $b$. Alice uses the homomorphic properties of this encoding to compute a noisy encoding of $ab+z$ for a random $z$ of her choice. Bob uses OT to retrieve only the non-noisy portions of the latter encoding. Note that the above unconditionally secure protocol can also be viewed as an instance of this general paradigm.

In more detail, suppose that $G$ is an $n\times k$ generating matrix of a linear code ${\cal C}\subset F^n$ whose minimal distance is bigger than $d$. This implies that an encoded message can be efficiently recovered from any $n-d$ coordinates of the encoding by solving a system of linear equations defined by the corresponding sub-matrix of $G$. Now, suppose that $G$ has the following intractability property: the distribution of $Gu+e$, where $u$ is a random message from $F^k$ whose first coordinate is $b$ and $e$ is a random noise vector of Hamming weight at most $d$, keeps $x$ semantically secure. (This follows, for instance, from the pseudorandomness of a noisy codeword in the code spanned by all but the first column of $G$.) Given such $G$ the protocol proceeds as follows. Bob sends to Alice $v=Gu+e$ as above, where $e$ is generated by first picking at random a subset $L\subset [n]$ of size $n-d$ and then picking $e_i$ at random for $i\not\in L$ and setting $e_i=0$ for $i\in L$. By assumption, $v$ keeps $b$ hidden from Alice. Alice now locally computes $v'=a\cdot v-Gz$, where $z$ is a random message in $F^k$. Restricted to the coordinates in $L$, this agrees with the encoding of a {\em random} message whose first coordinate is $ab-z_1$. Using the OT-oracle, Bob obtains from Alice only the coordinates of $v'$ with indices in $L$, from which it can decode and output $ab-z_1$. Alice outputs $z_1$.

The basic secure polynomial evaluation protocol from~\cite{NaorP06}, when restricted to degree-1 polynomials, essentially coincides with the above protocol when $\cal C$ is a Reed-Solomon code. The extension to general linear codes makes the underlying security assumption more conservative. Indeed, in contrast to Reed-Solomon codes, the problem of decoding {\em random} linear codes is believed to be intractable even for very low levels of noise.

In our actual protocols we will use several different distributions for picking the generating matrix $G$, and allow the noise distribution to depend on the particular choice of $G$ (rather than only on its minimal distance). In particular, for the case of general black-box rings we pick $G$ from a special class of codes for which decoding does not require inversion and yet the corresponding intractability assumption still seems viable.

Finally, in our most efficient code-based protocol we use Reed-Solomon codes as in~\cite{NaorP06}, but extend the above general template by letting Bob pack $t=\Omega(k)$ field elements $(b_1,\ldots,b_t)$ into the same codeword $v$. This variant of the construction does not apply to a general $G$, and relies on a special property of Reed-Solomon codes which was previously exploited in~\cite{FranklinYu92}.
This approach yields a protocol which realizes $t$ parallel instances of product-sharing by communicating only $O(t)$ field elements.

\superparagraph{Product-sharing from homomorphic encryption.}
Our last product-sharing protocol applies to rings of the form $\ZZ_m$ or $n\times n$ matrices over such rings and makes a standard use of homomorphic encryption. The only technicality that needs to be addressed is that the most useful homomorphic homomorphic encryption schemes do not allow to control the modulus $m$ but rather have this modulus generated by the key-generation algorithm. However, in the semi-honest model it is simple (via standard techniques) to emulate secure computation modulo $m$ via secure computation modulo any $M\gg m$.

\subsection{Further Discussion}
\label{sec-discuss}

\superparagraph{From the OT-hybrid model to the plain model}
An advantage of presenting our protocols in the OT-hybrid model is that they can be instantiated in a variety of models and under a variety of assumptions. For instance, using UC-secure OT protocols from~\cite{PeikertVaWa08,DamgardOT}, one can obtain efficient UC-secure instances of our protocols in the CRS model.
In the stand-alone model, one can implement these OTs by making a black-box use of homomorphic encryption~\cite{IshaiKLP06}. Thus, our protocols which make a black-box use of homomorphic encryption do not need to employ an additional OT primitive in the stand-alone model.

We finally note that our protocols requires only $O(k)$ OTs with security in the malicious model, independently of the circuit size; the remaining OT invocations can all be implemented in the semi-honest model, which can be done very efficiently using the technique of~\cite{IshaiKiNiPe03}. Furthermore, all the ``cryptographic'' work for implementing the OTs can be done off-line, before any inputs are available. We expect that in most natural instances of large-scale secure arithmetic computation, the cost of realizing the OTs will not form an efficiency bottleneck.

\superparagraph{Extensions.}
While we explicitly consider here only stateless arithmetic circuits, this model (as well as our results) can be readily generalized to allow stateful, reactive arithmetic computations whose secret state evolves by interacting with the parties.\footnote{An ideal functionality which formally captures such general reactive arithmetic computations was defined in~\cite{DamgardN03} (see also~\cite[Chapter~4]{Toft07thesis}) and referred to as an {\em arithmetic black-box} (ABB). All of our protocols for arithmetic circuits can be naturally extended to realize the ABB functionality.}

Another direction for extending the basic results has to do with the richness of the arithmetic computation model. Recall that the standard model of arithmetic circuits allows only to add, subtract, and multiply ring elements. While this provides a clean framework for the study of secure computation over black-box rings, many applications depend on other operations that cannot be efficiently expressed in this basic circuit model. For instance, when emulating arithmetic computation over the integers via computation over a (sufficiently large) finite field, one typically needs to check that the inputs comes from a given range.

As it turns out, reactive arithmetic computations are surprisingly powerful in this context, and can be used to obtain efficient secure realizations of useful ``non-arithmetic'' manipulations of the state, including decomposing a ring element into its bit-representation, equality testing, inversion, comparison, exponentiation, and others~\cite{DamgardFKNT06,Toft07thesis}.
These reductions enhance the power of the basic arithmetic model, and allow protocols to efficiently switch from one representation to another in computations that involve both boolean and arithmetic operations.

\subsection{Roadmap} We now briefly outline the
structure of the rest of this paper. Our basic definitions,
including those of black-box computational rings and our notion of
security in this context, are given in Section 2.  To achieve our
results (focusing on the two-party setting), recall that our overall
technical approach is to invoke~\cite{IshaiPrSa08}, which gives a
general blueprint for constructing efficient protocols by combining
an ``outer MPC protocol'' secure against active adversaries in the
honest majority setting, with an ``inner two-party protocol'' for
simple functionalities that need only be secure against
\emph{passive} adversaries.  We will give the details of this in
Section 5, but the bottom line (as discussed above) is that existing protocols (some with minor modifications) suffice for the outer MPC protocols, and all we
need to provide are efficient inner protocols secure against passive
adversaries.  Furthermore, since we are in the setting of passive
adversaries, the only functionality that we need the inner protocol
to compute is a basic ring multiplication function, at the end of
which the two parties should hold additive shares of the product of
their respective inputs.  To construct efficient protocols for this
basic functionality, we examine three approaches.  Our first two
approaches are based on ``noisy encodings" of various types, which
we define in Section 3, and the last approach is based on
homomorphic encryption.  The actual protocols (``inner two-party
protocols'') based on these three approaches are given in Section 4.

\section{Preliminaries}

\paragraph{Black-box rings and fields.}
A probabilistic oracle
\ringalgo
is said to be a
valid implementation of a finite ring \ring if it behaves as follows: it
takes as input one of the commands \add, \subtract, \mult, \sample and two
$m$ bit ``element identifiers'' (or none, in the case of \sample), and
returns a single $m$ bit string.
There is a one-to-one mapping
$\rep:\ring\hookrightarrow\{0,1\}^m$ such that for all $x,y\in\ring$
$\ringalgo(\operation,\rep(x),\rep(y))=\rep(x *_\ring y)$ where \operation
is one of \add, \subtract and \mult and $*_\ring$ is the ring operation
$+,-,$ or $\cdot$ respectively.
When an input is not from the range of $\rep$, the oracle outputs $\bot$. (In a
typical protocol, if a $\bot$ is ever encountered by an honest player, the protocol aborts.)
The output of
$\ringalgo(\sample)$ is $\rep(x)$ where $x$ will be drawn uniformly at random from \ring.
We will be interested in oracles of the kind that implements a {\em
family} of rings, of varying sizes. Such a function should take an
additional input \ringid to indicate which ring it is implementing.

\begin{definition}
A probabilistic oracle
\ringfamily is said to be a {\em concrete
ring family} (or simply a {\em ring family}) if, for all strings \ringid, the
oracle $\ringfamily(\ringid,\cdot)$ (i.e., with first input being fixed to
\ringid), is an implementation of some ring. This concrete ring will be denoted by $\ringfamily_\ringid$.
\end{definition}
Note that so far we have not placed any computability requirement on the oracle;
we only require a concrete mapping from ring elements to binary strings.
However, when considering computationally secure protocols we will
typically restrict the attention to
``efficient'' families of rings: we say \ringfamily is a  {\em
computationally efficient ring family} if it is a ring family that can be implemented by a
probabilistic polynomial time algorithm.

There are some special cases that we shall refer to:
\begin{enumerate}
\item Suppose that for all \ringid, we have that $\ringfamily_\ringid$ is a
ring with an identity for multiplication, 1.  Then, we call \ringfamily a
{\em ring family with inverse} if in addition to the other
operations, $\ringfamily(\ringid, \one)$ returns $\rep_\ringid(1)$ and
$\ringfamily(\ringid,\invert,\rep_\ringid(x))$ returns
$\rep_\ringid(x^{-1})$ if $x$ is a unit (i.e., has a unique left- and
right-inverse) and $\bot$ otherwise.
\item If \ringfamily is a ring family with inverse such that
for all \ringid the ring $\ringfamily_\ringid$ is a field, then we
say that \ringfamily is a {\em field family}.
\item We call a ring family with inverse \ringfamily a {\em
pseudo-field family}, if for all \ringid, all but negligible (in $|\ringid|$)
fraction of the elements in the ring $\ringfamily_\ringid$ are units.
\end{enumerate}

Some special families of rings we will be interested in, other than finite
fields, include rings of the form $\mathbb{Z}_m=\mathbb{Z}/m\mathbb{Z}$ for
a composite integer $m$ (namely, the ring of residue classes modulo $m$),
and rings of matrices over a finite field or ring. With an appropriate choice of
parameters, both of these families are in fact pseudo-fields. Note that a
concrete ring family \ringfamily for the rings of the form $\mathbb{Z}_m$
could use the binary representation of $m$ as the input \ringid; further the
elements in $\mathbb{Z}_m$ could be represented as $\lceil \log m
\rceil$-bit strings in a natural way. Of course, a different concrete ring
family for the same ring can use a different representation.

Finally, for notational convenience we assume that the length of all element
identifiers in $\ringfamily_\ringid$ is exactly $|\ringid|$. In particular,
the ring $\ringfamily_\ringid$ has at most $2^{|\ringid|}$ elements.

\paragraph{Arithmetic circuits.} An arithmetic circuit is a circuit (i.e., a
directed acyclic graph with the nodes labeled as input gates, output gates
or internal gates), in which the internal gates are labeled with a ring
operation: \add, \subtract or \mult.  (In addition, for fields, one often
considers the additional constant gate \one.) An arithmetic circuit \ckt can
be instantiated with any ring \ring. We denote by $\ckt^\ring$ the mapping (from a vector of ring elements to a vector of ring elements) defined in a natural way by instantiating \ckt with \ring. For a concrete ring family $\ringfamily$, we denote by $\ckt^\ringfamily$ the mapping which takes an \ringid and a vector of input identifiers and outputs the corresponding vector of output identifiers. (If any of the inputs is not a valid identifier, $\ckt^\ringfamily$ outputs $\bot$.)

In the context of multi-party computation, each input or output to such a
circuit is annotated to indicate which party (or parties) it ``belongs'' to.
Given such an annotated circuit \ckt and a concrete ring family \ringfamily,
we define the functionality \fckt\ringfamily to behave as follows:
\begin{itemize}
\item The functionality takes \ringid as a common (public) input, and
receives (private) inputs to \ckt from each party. It then
evaluates the function $\ckt^{\ringfamily}(\ringid,\text{inputs})$ using access to
\ringfamily, and provides the outputs to the parties.%
\footnote{\fckt{\ringfamily} can take \ringid as input from each party, and
ensure that all the parties agree on the same \ringid. Alternately, we can
restrict to environments which provide the same common input \ringid to all
parties.  In this case \ringid could be considered part of the specification of the
functionality, more appropriately written as \fcktid\ringfamily.
}
\end{itemize}

\paragraph{Protocols securely realizing arithmetic computations.} We follow
the standard UC-security framework \cite{Canetti05UC}. Informally, a
protocol \prot is said to securely realize a functionality \Ffunc if
there exists a PPT simulator \simx, such that
for all
(non-uniform PPT) adversaries \adv, and all (non-uniform PPT) environments \env which
interact with a set of parties and an adversary, the following two
scenarios are indistinguishable: the \real interaction where the parties run
the protocol \prot and the adversary is \adv; the \ideal interaction where
the parties communicate directly with the ideal functionality \Ffunc and the
adversary is $\simx^\adv$. Indistinguishability can either be statistical (in the case of unconditional security) or computational (in the case of computational security).
All parties, the adversary, the simulator, the
environment and the functionality get the security parameter $k$ as implicit
input.  Polynomial time computation, computational or statistical
indistinguishability and non-uniformity are defined with respect to this security parameter
$k$. However, since we don't impose an a-priori bound on the size of the inputs received from the environment, the running time of honest parties is bounded by a fixed polynomial in the total length of their inputs (rather than a fixed polynomial in $k$).

We distinguish between {\em static} corruption and {\em adaptive} corruption. In the latter case it
also makes a difference whether the protocols can erase their state (so that
a subsequent corruption will not have access to the erased information), or
no erasure is allowed. Our final protocols will have security against
adaptive\footnote{
One of the 
reasons for us to aim for adaptive security with erasure is that we will be relying on the main protocol compiler of~\cite{IshaiPrSa08}, as described informally in the Introduction and treated more formally in Section~\ref{sec:general}.  This compiler requires that the component protocols, the ``outer MPC protocol'' and the ``inner two-party protocol,'' both enjoy adaptive security -- the outer protocol must be adaptively secure in the model without erasures, but the inner protocol can be adaptively secure with erasures (in the OT-hybrid model).  Note that the conference version of~\cite{IshaiPrSa08} incorrectly claimed that the main protocol's proof of security works even when the inner protocol is only statically secure, but this does not seem to be the case.  However, this issue does not present any problems for us here, as we are easily able to modify our proposed ``inner'' protocols to achieve adaptive security with erasures using standard techniques, as detailed in Appendix~\ref{app:adaptive}.
}
corruption in the model that allows honest parties to erase their state information, but as an intermediate step, we will consider protocols which have security only against static corruption.

We shall consider protocols which make oracle access to a ring family \ringfamily.
For such a protocol we define its {\em arithmetic computation complexity} as
the number of oracle calls to \ringfamily. Similarly the {\em arithmetic
communication complexity} is defined as the number of ring-element labels in
the communication transcript. The arithmetic
computation (respectively communication) complexity of our protocols
will dominate the other computation steps in the protocol
execution (respectively, the number of other bits in the transcript).
Thus, the arithmetic complexity gives a good measure of efficiency for our protocols.

Note that while any computational implementation of the ring oracle necessarily requires the complexity to grow with the ring size, it is possible that the arithmetic complexity does not depend on the size of the ring at all.

We now define our main notion of secure arithmetic computation.
\begin{definition}
Let \ckt be an arithmetic circuit. A protocol \prot is said to be a {\em secure black-box realization of} \ckt-evaluation for a given set of ring
families if, for each \ringfamily in the set,
\begin{enumerate}
\item
$\prot^{\ringfamily}$ securely realizes $\fckt{\ringfamily}$,
and
\item the arithmetic (communication and computation) complexity of $\prot^\ringfamily$
is bounded by
some fixed polynomial in $k$ and $|\ringid|$
(independently of \ringfamily).
\end{enumerate}
\end{definition}
In the case of unconditional security we will quantify over the set of {\em all} ring families, whereas in the case of computational security we will typically quantify only over computationally efficient rings or fields.\footnote{This is needed only in the constructions which rely on concrete computational assumptions. A computationally-unbounded ring oracle can be used by the adversary to break the underlying assumption. }
In both cases, the efficiency requirement on \prot rules out the option of using a brute-force approach to emulate the ring oracle by a boolean circuit.

We remark that our constructions will achieve a stronger notion of 
security, as the simulator used to establish the security in
item (1) above will not depend on \ringfamily.
A bit more precisely, the stronger definition is quantified as follows: there exists a simulator such that for all adversaries, ring families, and environments, the ideal process and the real process are indistinguishable.
 For simplicity however we
phrase our definition as above which does allow different simulators for
different \ringfamily.

\section{Noisy Encodings}
\label{sec:noisyenc}

A central tool for our main protocols is a noisy encoding of elements in a
ring or a field.  In general this encoding consists of encoding a
randomly padded message with a (possibly randomly chosen) linear code, and
adding noise to the codeword obtained. The encodings will be such that, with
some information regarding the noise, decoding (of a codeword derived from
the noisy codeword) is possible, but otherwise the noisy codeword hides the
message. The latter will typically be a computational assumption, for
parameters of interest to us.

We shall use two kinds of encodings for our basic protocols in
\sectionref{passive}. The first of these encodings has a statistical hiding
property which leads to a statistically secure protocol (in the \OT-hybrid
model). The other kind of encoding we
use (described in \sectionref{lincode}) is hiding only under computational assumptions. In fact, we provide a
general template for such encodings and instantiate it variously, leading to different concrete computational assumptions. 

\subsection{A Statistically Hiding Noisy Encoding}
\label{sec:stathiding}

\begin{itemize}
\item {\bf Encoding of $x$, $\statencode(\ringid,x)$}: Here $x\in\ringfamily_\ringid$; $n$ is
a parameter of the encoding.
	\begin{itemize}
	\item Denote $\ringfamily_\ringid$ by \ring.
	\item Pick a ``pattern'' $\sigma\in\{0,1\}^n$.
	\item Pick  a random vector $u\in\ring^n$ conditioned on $\sum_{i=1}^n u_i = x$.
	\item Pick a pair of random vectors $(v^0,v^1)\in\ring^n\times\ring^n$,
	conditioned on	$v^{\sigma_i}_i = u_i$. That is, the vector $u$ is
	``hidden'' in the pair of vectors $v^0$ and $v^1$ according to the pattern
	$\sigma$.
	\item Output $(v^0,v^1,\sigma)$.
	\end{itemize}
\end{itemize}

The encoding could be seen as consisting of two parts $(v^0,v^1)$ and $\sigma$, where
the latter is information that will allow one to decode this code.
For $x\in\ring$, let $\scode_x^{\ring,n}$ denote the distribution of
the first part of the encoding $\statencode(\ringid,x)$, namely $(v^0,v^1)$.

This simple encoding has the useful property that it statistically
hides $x$ when the decoding information $\sigma$ is removed. The proof of
this fact makes use of the Leftover Hash Lemma~\cite{ImpagliazzoLL89} (similarly to previous uses of this lemma
in~\cite{ImpagliazzoNaor96,IshaiKOS06}).

\begin{lemma}
\label{lem:stathiding}
Let $n>\log|\ring| + k$. Then, for all $x\in\ring$, the statistical distance between
the distribution of $\scode_x^{\ring,n}$ and the uniform distribution over
$\ring^n\times\ring^n$ is $2^{-\Omega(k)}$.
\end{lemma}
\begin{proof}
Consider the hash function family \hash that consists of functions
$H_{v^0,v^1} : \{0,1\}^n \rightarrow \ring$, where
$(v^0,v^1)\in\ring^n\times\ring^n$, defined as $H_{v^0,v^1}(\sigma) :=
\sum_i v^{\sigma_i}_i$.  It is easily verified that this is a 2-universal
hash function family.  Then, by the Leftover Hash Lemma,
\[\frac1{|\hash|} \sum_{H\in\hash} \statdiff(H(\unif_{\{0,1\}^n}),\unif_\ring)
    = 2^{-\Omega(n-\log|\ring|)},\]
where $\unif_{\{0,1\}^n}$ stands for the uniform distribution over $\{0,1\}^n$
and $\statdiff$ denotes the statistical difference between two
distributions.

To prove the lemma we make use also of the following symmetry between all
the possible outcomes of the hash functions: There is a family of
permutations on \hash, $\{\pi_\alpha|\alpha\in\ring\}$ such that for
all $z\in\ring$, $\pr{z|H}=\pr{z+\alpha|\pi_\alpha(H)}$ (where $\pr{z|H}$ is
a shorthand for $\Pr_{\sigma\from\{0,1\}^n}[H(\sigma)=z]$).  In particular
we can set $\pi_\alpha(H_{v^0,v^1}) := H_{u^0,u^1}$ where $u^0$
(respectively $u^1$) is identical to $v^0$ (respectively $v^1$) except for
the first co-ordinate which differs by $\alpha$:
$u^0_1-v^0_1 = u^1_1-v^1_1 = \alpha$.
Then,
\begin{align*}
\frac1{|\hash|} \sum_{H\in\hash} \statdiff(H(\unif_{\{0,1\}^n}),\unif_\ring)
 &=\frac1{|\hash|} \frac12 \sum_{H\in\hash} \sum_{z\in\ring}
     \left( |\pr{z|H}-\frac1{|\ring|}| \right) \\
 &=\frac1{|\hash|} \frac12 \sum_{z\in\ring} \sum_{H\in\hash}
     \left( |\pr{x|\pi_{x-z}(H)}-\frac1{|\ring|}| \right)  \\
 &=\frac{|\ring|}{|\hash|} \frac12 \sum_{H\in\hash}
     \left( |\pr{x|H}-\frac1{|\ring|}| \right)  \;\;\;\text{because $\pi_{x-z}$ is a permutation}\\
 &=\frac12 \sum_{H\in\hash}
     \left( |\pr{H|x}-\frac1{|\hash|}| \right) \;\;\;\text{because with
     $\pr{H}=\frac1{|\hash|}$, $\pr{x}=\frac1{|\ring|}$.}
\end{align*}
Note that the last expression is indeed the statistical difference between
$\scode_x^{\ring,n}$ and $\unif_{\ring^n\times\ring^n}$. To complete the proof note that we
have already bounded the first quantity by $2^{-\Omega(n-\log|\ring|)}$.
\end{proof}

\subsection{Linear Code Based Encodings}
\label{sec:lincode}

We describe an abstract noisy encoding scheme for a ring family \ringfamily.
The encoding scheme is specified using  {\em a code generation algorithm \codegen}:
\begin{itemize}
\item $\codegen$ is a randomized
algorithm such that $\codegen^\ringfamily(\ringid)$ outputs $(G,H,L)$ where
$G$ is an $n\times k$ matrix,  $L\subseteq [n]$, $|L|=\ell$ and $H$ is
another matrix. We note that only $G$ and $L$ will be used in the noisy encoding process; $H$ will be useful
in describing the decoding process. 
\end{itemize}
Here $k$ is the security parameter as well as the code dimension, and $n(k)$ (code length) and $\ell(k)$ (number of coordinates {\em without} noise) are parameters of \codegen. In our instantiations $n$ will be a constant multiple of $k$ and in most cases we will have $\ell=k$.

Let \ringfamily be a ring family and $\ring=\ringfamily_\ringid$ from some \ringid.
Given \codegen, a parameter $t(k)\le k$ (number of ring elements to be encoded, $t=1$ by default), and $x\in\ring^t$, we
define a distribution $\lincode{\codegen}{\ring}t{x}$, as that of the public
output in the following encoding process:
\begin{itemize}
\item {\em Encoding $\encode{\codegen}{\ringfamily}t{\ringid,x}$:}
\begin{itemize}
\item Input: $x=(x_1,\ldots,x_t)\in\ring^t$.
\item Let $(G,L,H)\from \codegen^\ringfamily(\ringid)$
\item Pick a random vector $u\in\ring^k$ conditioned on $u_i = x_i$ for $i=1,\ldots,t$
(i.e., $u$ is $x$ padded with $k-t$ random elements). Compute $Gu \in\ring^n$.
\item Pick a random vector
$v\from\ring^n$, conditioned on $v_i := (Gu)_i$ for $i\in L$.
\item Let the private output be $(G,L,H,v)$ and the public output be $(G,v)$
(where each ring element is represented as a bit string
obtained by the mapping $\rep$ used by \ringfamily).
\end{itemize}
\end{itemize}

The matrix $H$ is not used in the encoding above, but will be required for a
decoding procedure that our protocols will involve. In our
main instantiations $H$ can be readily derived from $G$ and $L$. But we include
$H$ explicitly in the outcome of \codegen, because in some cases it is
possible to obtain efficiency gains if $(G,H,L)$ are sampled together.
We sketch one such case when we describe ``Ring code based encoding''
in \sectionref{encoding-insts}.

\begin{assumption}
\label{asm:lincode-generic}
{\bf (Generic version, for a given \codegen, \ringfamily and $t(k)$.)} For
all sequences $\{(\ringid_k,x_k,y_k)\}_k$, let $\ring_k = \ringfamily_{\ringid_k}$,
and suppose $x_k,y_k \in \ring_k^{t(k)}$. Then the ensembles
$\{\lincode{\codegen}{\ring_k}t{x}\}_k$ and
$\{\lincode{\codegen}{\ring_k}t{y}\}_k$ are computationally
indistinguishable.
\end{assumption}

For the sake of reference to some previously studied assumptions, we also
define a simpler (but stronger) generic assumption, which implies the above version:
\begin{assumption}
\label{asm:lincode-generic-pr}
{\bf (Generic pseudorandomness version, for a given \codegen and \ringfamily.)}
For any sequence $\{\ringid_k\}_k$, let $\ring_k = \ringfamily_{\ringid_k}$.
Then the ensembles
$\{\lincode{\codegen}{\ring_k}t{0^{t(k)}}\}_k$
and
$\{ \left( G \from \codegen^{\ring_k}, v \from \ring_k^n \right) \}_k$
are computationally indistinguishable.
\end{assumption}

\subsubsection{Instantiations of the Encoding}
\label{sec:encoding-insts}
The above generic encoding scheme can be instantiated by specifying a code
generation algorithm \codegen, a ring family, and the parameter $t(k)$ which specifies the
length of the input to be encoded. We consider three such instantiations.

\paragraph{Random code based instantiation.} Our first instantiation of
the generic encoding has
$t(k)=1$ and uses a code generation algorithm \randcodegen based on a random
linear code.  Here the ring family is any field family \fieldfamily.
$\randcodegen^\fieldfamily(\ringid)$ works as follows:
\begin{itemize}
\item Let $k=|\ringid|$. Let $n=2k$ and $\ell=k$. Denote
$\fieldfamily_\ringid$ by \field.
\item Pick a random	$n\times k$ matrix $G\from\field^{n\times k}$.
\item Pick a random subset $L\subseteq [n]$, $|L|=k$, such that the $k\times
k$	submatrix $G|_L$ is non-singular, where $G|_L$ consists of those rows in
$G$ whose indices are in $L$.%
\footnote{For efficiency of \codegen, it is enough to try random subsets
$L\subseteq [n]$ and check if $G|_L$ is non-singular; in the unlikely event
that no $L$ is found in $k$ trials, \codegen can replace $G$ with an
arbitrary matrix with a $k\times k$ identity matrix in the first $k$ rows.}
\item Let $H$ be the $k\times k$ matrix such that $HG|_L = I$, the
$k\times k$ identity matrix. (This $H$ will be used in our protocol
constructions.)
\end{itemize}

The following variants of this instantiation are also interesting:
\begin{itemize}
\item Instead of choosing $n(k)=2k$, we can choose
$n(k)>2k+\log^c|\field|$ for some $c<1$ (say $c=\frac12$). By
choosing a larger $n$ we essentially weaken the required assumption.
(We remark that the case of $n(k)>\log|\field|$ is not of much interest to us here,
because then our construction which employs this assumption is bettered by
our unconditional construction.)
\item The above encoding can be directly used with a pseudo-field family
instead of a field family. Note that the invertibility of elements was used
in deriving $H$, but in a pseudo-field, except with negligible probability
this derivation will still be possible.
\end{itemize}

\paragraph{Ring code based instantiation.}
Our next instantiation also has $t(k)=1$. It uses a code generation
algorithm \ringcodegen that works with any arbitrary ring family (not just
fields). But for simplicity we will assume that the ring has a
multiplicative identity $1$.%
\footnote{Rings which do not have $1$ can be embedded into a ring of double
the size which does have $1$, by including new elements $a+1$
for every element $a$ in the original ring, and setting $1+1=0$.}
Here again in the noisy encoding we will use $t=1$.
$\ringcodegen^\ringfamily(\ringid)$ works follows.
\begin{itemize}
\item Let $k=|\ringid|$. Let $n=2k$ and $\ell=k$. Denote
$\ringfamily_\ringid$ by \ring.
\item Pick two $k\times k$ random matrices
$A$ and $B$ with elements from \ring, conditioned on them being
upper triangular  and having $1$ in the main diagonal.
Let $G$ be the $2k\times k$ matrix
$\left[\begin{smallmatrix}A\\B\end{smallmatrix}\right]$.
\item Define $L$ as follows.
Let $L=\{a_1,\ldots,a_k\}$ where $a_i = i$ or $k+i$ uniformly at random. (That is $a_i$ indices
the $i$-th row in either $A$ or $B$.)
\item Note that $G|_L$ is an upper triangular matrix with 1 in the main
diagonal. It is easy to compute an upper triangular matrix $H$ (also with
1 in the main diagonal) using only the ring operations on elements in
$G|_L$ such that $H G|_L = I$.
\end{itemize}
Here, instead of choosing two matrices, we could choose several, to make the
resulting assumption weaker at the expense of increasing $n$.

We point out an alternate encoding which would also work with arbitrary
rings. 
One can construct $G|_L$ and $H$ such that $HG|_L = I$ simultaneously by
taking a two {\em opposite} random walks in the special linear group
${\mathrm{SL}}(n,\ring)$ (i.e., the group of $n\times n$ matrices over the
ring \ring, with determinant 1), where each step in the walk consists of
adding or subtracting a row from another row, or a column from another
column; in the ``opposite'' walk, the step corresponding to an addition has
a subtraction, and the step corresponding to subtraction has an addition.
The random walks start from the identity matrix, and will be long enough
for the generated matrices to have sufficient entropy.
Note that in such a scheme, we need to rely on the code generation algorithm
to simultaneously sample $(G,L,H)$, rather than output just $(G,L)$, 
because matrix inversion is not necessarily easy for all rings.

\paragraph{Reed-Solomon code based instantiation.}
In our third instantiation of the generic encoding, we will have
$t(k)$ to be a constant fraction of $k$. The code generation algorithm
\rscodegen is
based on the Reed-Solomon code, and will work with any sufficiently large
field family \fieldfamily.
$\rscodegen^\fieldfamily(\ringid)$ works as follows:
\begin{itemize}

\item Let $k=|\ringid|$. Let $n=ck$, for a sufficiently large
constant%
\footnote{We require $c>4$ so that Assumption~\ref{asm:lincode-specific}(c)
will not be broken by known list-decoding algorithms for Reed-Solomon codes.
$c=8$ may be a safe choice, with larger values of $c$ being more
conservative.}
$c>4$, and $\ell=2k-1$. Denote $\fieldfamily_\ringid$ by \field.

\item Pick distinct points $\x{i}\in\field$ for $i=1,\ldots,k$, and
$\y{i}\in\field$, for $i=1,\ldots,n$ uniformly at random.

\item Define the $n\times k$ matrix $G$ so that it extrapolates a
degree $k-1$ polynomial, given by its value at the $k$ points
\x{i}, to the $n$ evaluation points \y{i}. That is, $G$ is such that for any
$u\in\field^k$, $(Gu)_i = P(\y{i})$ for $i=1,\ldots,n$, where $P$ is the
unique degree $k-1$ polynomial such that $P(\x{i})=u_i$ for $i=1,\ldots,k$.

\item Pick $L\subseteq[n]$ with $|L|=\ell=2k-1$ at random.

\item Let $H$ be the $k\times 2k-1$ matrix such that $(Hv_L)_i = Q(\x{i})$,
where $Q$ is the unique degree $2(k-1)$ polynomial such that $Q(\y{j})=v_j$
for all $j\in L$.

\end{itemize}

\subsubsection{Instantiations of Assumption~\ref{asm:lincode-generic}}
Each of the above instantiations of the encoding leads to a corresponding
instantiation of Assumption~\ref{asm:lincode-generic}. For the sake of
clarity we collect these assumptions below.

\begin{assumption}
\label{asm:lincode-specific}
\begin{enumerate}
\item[(a)]
{\bf [For \randcodegen, with $t(k)=1$.]} For any computationally efficient
field family \fieldfamily, for all sequences $\{(\ringid_k,x_k,y_k)\}_k$,
let $\field_k = \fieldfamily_{\ringid_k}$, and suppose $x_k,y_k \in \field_k$.
Then the ensembles
$\{\lincode{\randcodegen}{\field_k}1{x}\}_k$ and
$\{\lincode{\randcodegen}{\field_k}1{y}\}_k$ are computationally
indistinguishable.
\item[(b)]
{\bf [For \ringcodegen, with $t(k)=1$.]} For any computationally efficient
ring family \ringfamily, for all sequences $\{(\ringid_k,x_k,y_k)\}_k$,
let $\ring_k = \ringfamily_{\ringid_k}$, and suppose $x_k,y_k \in \ring_k$.
Then the ensembles
$\{\lincode{\ringcodegen}{\ring_k}1{x}\}_k$ and
$\{\lincode{\ringcodegen}{\ring_k}1{y}\}_k$ are computationally
indistinguishable.
\item[(c)]
{\bf [For \rscodegen, with $t(k)=k/2$.]}%
\footnote{We can make the assumption weaker by choosing smaller values of
$t$, or larger values of $n$ in \rscodegen.}
For any computationally efficient
field family \fieldfamily, for all sequences $\{(\ringid_k,x_k,y_k)\}_k$,
let $\field_k = \fieldfamily_{\ringid_k}$, and suppose $x_k,y_k \in \field_k^{t(k)}$.
Then the ensembles
$\{\lincode{\randcodegen}{\field_k}t{x}\}_k$ and
$\{\lincode{\randcodegen}{\field_k}t{y}\}_k$ are computationally
indistinguishable, for $t\le k/2$.
\end{enumerate}
\end{assumption}

\section{Product-Sharing Secure Against Passive Corruption}
\label{sec:passive}

In this section we consider the basic two-party functionality \fmult
described below
\begin{itemize}
\item
\alice sends $a \in \ring$ and \bob sends $b\in\ring$ to \fmult.
\item \fmult samples two random elements $z^\alice, z^\bob \in \ring$ such that
$z^\alice+z^\bob = ab$, and gives $z^\alice$ to \alice and $z^\bob$ to \bob.
\end{itemize}
When we want to explicitly refer to the ring in which the computation
takes place we will write the functionality as $\fmult^\ring$.

We present three protocols based on noisy encodings, with increasing
efficiency, but using stronger assumptions, in the \OT-hybrid model for this
functionality (some of which are restricted to when \ring is a field).  We
then present two protocols based on homomorphic encryption. These protocols
are secure only against {\em static} passive corruption. In
Appendix~\ref{app:adaptive} we present a general transformation, that
applies to a class of protocols covering all our above protocols, to obtain
protocols that are secure against adaptive passive corruption, with
erasures.

\subsection{A Basic Protocol with Statistical Security}
\label{sec:basic-stat}

\begin{itemize}
\item {\bf Protocol \pmultbasicx.} \alice holds $a \in \ring$ and \bob holds
$b \in \ring$.
	\begin{itemize}
	\item \bob randomly encodes $b$ as specified in \sectionref{stathiding}.
	i.e., let $(v^0,v^1,\sigma)\from\statencode(\ringid,b)$. Then $\sum_i v^{\sigma_i} = b$.
	
	\item \bob sends $(v^0_i,v^1_i)$ (for $i=1,\ldots,n$) to \alice.

	\item \alice picks a random vector $t\in\ring^n$ and sets $z^\alice =\sum_{i=1}^n t_i$;
	she computes $w^0_i=a v^0_i - t_i$ and $w^1_i = a v^1_i - t_i$.

	\item \alice and \bob engage in $n$ instances of $2\choose1$ \OT, where in
	the $i^\text{th}$ instance \alice's inputs are $(w^0_i,w^1_i)$ and \bob's
	input is $\sigma_i$. \bob receives $w^{\sigma_i}_i$.

	\item \alice outputs $z^\alice$. \bob outputs the sum of all the $n$
	elements he received above: i.e., \bob outputs
	\[z^\bob := \sum_i w^{\sigma_i}_i = \sum_i \left( a v^{\sigma_i} - t_i \right) = a b -
	z^\alice. \]
	\end{itemize}
\end{itemize}

We will pick $n>\log(|\ring|)+k$. Then we have the following result.

\begin{lemma}
Suppose $n>\log(|\ring|) + k$.
Then protocol \pmultbasicx securely realizes \fmult 
against static passive corruption. The security is statistical.
\end{lemma}

\begin{sketch}
When \bob is corrupted, it is easy to construct a simulator to
obtain perfect security.
The more interesting case is when \alice is corrupted. Then the
simulator \simx behaves as follows.
\begin{itemize}
\item Send \alice's input to \fmult and obtain $z^\alice$ in response.
\item Set $t\in\ring^n$ in \alice's random tape such that $\sum_i t_i =
z^\alice$. Note that \alice's output will then be $z^\alice$.
\item Sample an element $\alpha\from\ring$ and run the honest program for
	\bob using this input. The only message produced by the simulation
	is a pair of vectors $(v^0,v^1)$.
\end{itemize}
By Lemma~\ref{lem:stathiding}, the message produced by the simulator is statistically close to the
message produced by \bob in the real execution (both being statistically
close to the uniform distribution over
$\ring^n\times\ring^n$), and the simulation is statistically
indistinguishable from a
real execution.
\end{sketch}

\subsection{Basic Protocol Using Linear Codes for Rings}
\label{sec:basic-lincode}

We improve on the efficiency of the protocol in \sectionref{basic-stat} by
depending on computational assumptions regarding linear codes.  One
advantage of the protocol in this section is that it does not explicitly
depend on the size of the underlying ring. Restricted to fields, this
construction can use the code generation \randcodegen; for arbitrary rings
with unity, the construction can use \ringcodegen. Note that both coding
schemes generate $(G,L,H)$ such that $HG|_L = I$, which is what the protocol
depends on. It uses these codes in a noisy encoding with $t=1$.

\begin{itemize}
\item {\bf Protocol \pmultlcx.} \alice holds $a \in \ring$ and \bob holds
$b \in \ring$.
    \begin{itemize}
    \item \bob randomly encodes $b$ using $\encode{\codegen}{\ring}1{b}$
        to get $(G,H,L,v)$ as the private output. (Note that $t=1$ in
        the encoding, and $HG|_L = I$.)
    \item \bob sends $(G,v)$ to \alice.

    \item \alice picks a random vector $x\in\ring^k$ and sets $w = av-Gx$.

    \item \alice and \bob engage in an $n\choose{k}$-\OT where \alice's
    inputs are $(w_1,\ldots,w_n)$ and \bob's input is $L$. \bob receives
    $w_i$ for $i\in L$. (Recall that when considering passive corruption,
    an $n\choose{k}$-\OT maybe implemented using $n$ instances of
    $2\choose1$-\OT. Here \OT is a string-OT and the inputs are labels for
    the ring elements.)

    \item \alice outputs $z^\alice := x_1$, the first co-ordinate of $x$.
    \bob outputs $z^\bob := \left(H w_L\right)_1= a b - x_1$.
    \end{itemize}
\end{itemize}

\begin{lemma}
If Assumption~\ref{asm:lincode-generic} holds for a code generation scheme
\codegen, with $t=1$, then
Protocol~\pmultlcx securely realizes
\fmult, against static passive corruption.
\end{lemma}

\begin{sketch}
The interesting case is when \alice is corrupt and \bob is honest.
Then the simulator \simx behaves as follows.
\begin{itemize}
\item Send \alice's input to \fmult and obtain $z^\alice$ in response.
\item Set $x\in\ring^n$ in \alice's random tape conditioned on $x_1 = z^\alice$.
Note that \alice's output will then be $z^\alice$.
\item Sample an element $\alpha\in\ring$ and run the honest program for
    \bob using this input. The only message produced by the simulation
    is the pair $(G,v)$.
\end{itemize}

$(G,v)$ is the only message output by (simulated) \bob
in the (simulated) protocol. In the real execution this
message is distributed according to
$\lincode{\codegen}{\ring}1{b}$
whereas in the simulation it is distributed according to
$\lincode{\codegen}{\ring}1{\alpha}$.
By the assumption in the lemma, we conclude that these
two distributions are indistinguishable (even if $b$ and
$\alpha$ are known), and hence
the view of the environment in the real execution is
indistinguishable from that in the simulated execution.
\end{sketch}

\subsection{Amortization using Packed Encoding}
\label{sec:basic-packed}

In this section we provide a passive-secure protocol in the
\OT-hybrid model for {\em multiple instances of} the basic two-party
functionality \fmult.  That is, we realize the two-party
functionality \fmultt which takes as inputs $\at \in \field^t$ and
$\bt \in \field^t$, and outputs random vectors $\zt^\alice$ and
$\zt^\bob$ to \alice and \bob respectively, such that $\zt^\alice +
\zt^\bob = \at \bt := (a_1b_1,\ldots,a_tb_t)$ (note that
multiplication in $\field^t$ refers to coordinate-wise
multiplication).

We use the noisy encoding scheme with the code generation algorithm
\rscodegen. We shall choose $t=k/2$.

\begin{itemize}
\item {\bf Protocol \pmulttx.}  \alice holds $\at = (a_1,\ldots,a_t) \in \field^t$ and
\bob holds $\bt = (b_1,\ldots,b_t) \in \field^t$.
    \begin{itemize}
    \item \bob randomly encodes $x$ using $\encode{\rscodegen}{\ring}t{x}$
        to get $(G,H,L,\vn)$ as the private output.
    \item \bob sends $G$ and $\vn = (v_1,\ldots,v_n) \in\field^n$ to
        \alice. Recall that for some degree $k-1$ polynomial \pb, $v_i := \pb(\y{i})$ for $i\in L$
        (and $v_i$ is a random field element if $i\not\in L$).

    \item Note that the points \y{i} and \x{i} are implicitly specified by
        $G$. \alice picks a random degree $k-1$ polynomial \pa such that
        $\pa(\x{i}) = a_i$ for $i=1,\ldots,t$, and also a random degree
        $2(k-1)$ polynomial \pc.  \alice computes $w_i := \pa(\y{i})v_i -
        \pc(\y{i})$ for $i=1,\ldots,n$.

    \item \alice and \bob engage in a $n\choose{2k-1}$ \OT, where \alice's
        inputs are $(w_1,\ldots,w_n)$ and \bob's input is $L$. \bob receives
        $w_i$ for $i\in L$.

    \item \bob computes $Hw|_L$. Note that then $(Hw|_L)_i = \pz(\x{i})$ where
    \pz is the unique degree $2(k-1)$ polynomial \pz such that $\pz(\y{i}) = w_i$ for $i\in L$.

    \item \alice sets $z^\alice_i := \pc(\x{i})$ for $i=1,\ldots,t$.
        and \bob sets $z^\bob_i := \pz(\x{i})$ for $i=1,\ldots,t$.

        Note that (if \alice and \bob are honest), \pz is the degree $2(k-1)$
        polynomial $\pa\pb - \pc$, and hence $z^\alice_i + z^\bob_i =
        \pa(\x{i})\pb(\x{i}) = a_i b_i$.

    \item \alice outputs $\zt^\alice := (z^\alice_1,\ldots,z^\alice_t)$ and
        \bob outputs $\zt^\bob := (z^\bob_1,\ldots,z^\bob_t)$.

    \end{itemize}
\end{itemize}

\paragraph{Remark about computational efficiency.}
The computational complexity of Protocol \pmulttx (ignoring the use
of \OT) is dominated by the evaluation and interpolation of
polynomials (note that the matrices $G$ and $H$ can be stored in an
implicit form just by storing the points \y{i} and \x{i}). As such,
in general the complexity would be $O(k \log^2 k)$ for randomly
chosen evaluation points~\cite{vzGG99}. We note, however, that this
complexity can be reduced to $O(k \log k)$ by a more careful
selection of evaluation points~\cite{vzGG99}, at the expense of
having to assume that Assumption ~\ref{asm:lincode-specific}(c)
holds also with respect to this specific choice of evaluation
points.

\begin{lemma}
If Assumption~\ref{asm:lincode-specific}(c) holds, then
Protocol~\pmulttx securely realizes \fmultt,
against static passive
corruption.
\end{lemma}

\begin{sketch}
The interesting case is when \alice is corrupt and \bob is honest.
Then the simulator \simx behaves as follows.
\begin{itemize}
\item Send \alice's input \at to \fmultt and obtain $\zt^\alice$ in response.
\item Set \alice's random tape so that she picks \pc such that
$\pc(\x{i})=z^\alice_i$ for $i=1,\ldots,t$.
Note that \alice's output will then be $\zt^\alice$.
\item Sample $\alphat\in\field^t$ and run the honest program for
    \bob using this input. The only message produced by the simulation
    is the vector \vn.
\end{itemize}
Indistinguishability of the simulation follows
because of the assumption in the lemma:
given \alphat and \bt,
$\lincode{\rscodegen}{\ring}t{\alphat}$
and
$\lincode{\rscodegen}{\ring}t{\bt}$
are computationally indistinguishable.
\end{sketch}

\subsection{Protocols based on Homomorphic Encryption}
\label{sec:homo}

In this section, we construct protocols (secure against passive
adversaries) for the basic two-party functionality \fmult, based on
homomorphic encryption. Since we work in the context of rings, by
homomorphic encryption (informally speaking), we mean an encryption
scheme where it is possible to both: (1) given encryptions of two
ring elements $x$ and $y$, it is possible to generate an encryption
of $x+y$; and (2) given a ring element $\alpha$ and an encryption of
a ring element $x$, it is possible to generate an encryption of
$\alpha x$.  It is important to stress two points:
\begin{itemize}
\item Any encryption scheme that is \emph{group-homomorphic} for the standard representation of the (additive) group $\mathbb{Z}_m$ is immediately homomorphic in our sense with respect to the \emph{ring} $\mathbb{Z}_m$.
\item As such, our notion of homomorphic encryption, even though it is defined in the context of rings, \emph{is different from and should not be confused with} the notion of ``fully'' or ``doubly'' homomorphic encryption.  In particular, we do not require that given encryptions of two ring elements $x$ and $y$, it is possible to generate an encryption of $x \cdot y$, where $\cdot$ is the ring multiplication operation.
\end{itemize}

Note that while most homomorphic encryption schemes
from the literature fit this definition (since they are group-homomorphic for the standard representation of the (additive) group $\mathbb{Z}_m$), some do not; for example,
the El Gamal encryption scheme is group-homomorphic for a subgroup of $\mathbb{Z}_p^*$, but there does not seem to be any ring structure for which El Gamal encryption would be homomorphic in our sense\footnote{
Since $\mathbb{Z}_p^*$ is cyclic, it can be associated with the ring $\mathbb{Z}_{p-1}$; however there does not seem to be any computationally efficient way to consider El Gamal encryption to be homomorphic for any nontrivial subring of this ring, as it would seem to require computing discrete logarithms in $\mathbb{Z}_p^*$ or its subgroups.}.

Furthermore, we consider two types of homomorphic encryption
schemes. Informally speaking, the issue that separates these two
types of homomorphic encryption schemes is whether the ring
underlying the homomorphic encryption scheme can be specified
beforehand (which we call a ``controlled ring" scheme), or whether
it is determined by the key generation algorithm (which we call an
``uncontrolled ring"). For example, the key
generation algorithm of the classic Goldwasser-Micali encryption
scheme~\cite{GoldwasserMi84} based on quadratic residuosity always
produces keys for a $\mathbb{Z}_2$-homomorphic encryption scheme,
and is thus a ``controlled ring" scheme.  Note that by considering
higher residuosity classes, Benaloh~\cite{Benaloh87thesis} similarly
constructs ``controlled ring" homomorphic encryption schemes for the
rings $\mathbb{Z}_p$, where $p$ is a polynomially bounded (small)
prime number.  On the other hand, schemes like the Paillier
cryptosystem~\cite{Paillier99} are homomorphic with respect to the
ring $\mathbb{Z}_n$, where $n$ is a randomly chosen product of two
large primes chosen at the time of key generation; $n$ cannot be
specified ahead of time. Thus, the Paillier scheme is an example of
an ``uncontrolled ring" homomorphic encryption scheme.

We first describe formally what we call ``controlled ring"
homomorphic encryption:

\begin{definition}
A \emph{``controlled ring" homomorphic encryption} scheme
corresponding to a concrete ring family \ringfamily is a
tuple of algorithms $(G, E, D, C)$, such that:

\begin{enumerate}
\item $(G,E,D)$ is a semantically secure public-key encryption scheme, except that the algorithm $G$
takes as input both $1^k$ and \ringid, and the set of values that
can be encrypted using the public-key output by $G$ are the elements
of $\ringfamily_\ringid$.

\item For any $x_1, x_2 \in \ringfamily_\ringid$, given $(pk,sk) \leftarrow G(1^k,\ringid)$ and two
ciphertexts $c_1 = E(pk, x_1)$ and $c_2 = E(pk, x_2)$, we have that
$C(pk, c_1, c_2)$ outputs a distribution whose statistical distance
to the distribution $E(pk, x_1 + x_2)$ is negligible in $k$.

\item For any $x, \alpha \in \ringfamily_\ringid$, given $(pk,sk) \leftarrow G(1^k,\ringid)$ and a
ciphertext $c = E(pk, x)$, we have that $C(pk, c, \alpha)$ outputs a
distribution whose statistical distance to the distribution $E(pk,
x\cdot\alpha)$ is negligible in $k$.
\end{enumerate}

\end{definition}
Such controlled ring homomorphic encryption schemes immediately give
rise to a protocol for our basic two-party functionality \fmult, as
we now demonstrate.

\begin{itemize}
\item {\bf Protocol $\theta$.} \alice holds $a \in \ringfamily_\ringid$ and \bob holds $b \in \ringfamily_\ringid$.

    \begin{itemize}

    \item (Initialization)  \alice runs $G(1^k, \ringid)$ to obtain $(pk,sk)$.  This is done only once, as  the same public key can be used as many times as necessary.

    \item \alice computes $c = E(pk, a)$, and sends $c$ to \bob.

    \item \bob chooses $r \in \ringfamily_\ringid$ at random, computes $c' = E(pk,r)$, and then computes $c'' = C(pk, C(pk, c, b), c')$ and sends $c''$ to \alice.  Note that $c''$ is an encryption of $ab+r$.  \bob outputs $-r$.

    \item \alice computes $v = D(sk,c'')$, and outputs $v$.
    \end{itemize}
\end{itemize}

The correctness and privacy properties of this protocol (against
passive corruptions) follow immediately from the definition of
controlled ring homomorphic encryption.

As mentioned above, unfortunately many known homomorphic encryption
schemes do not allow complete control over the ring underlying the
homomorphic encryption scheme, and so they do not satisfy the
definition of controlled ring homomorphic encryption schemes.  We
deal with these types of homomorphic encryption schemes separately
below.

\begin{definition}
An \emph{``uncontrolled ring" homomorphic encryption}
scheme corresponding to a concrete ring family \ringfamily
is a tuple of algorithms $(G, E, D, C)$, such that:

\begin{enumerate}
\item $(G,E,D)$ is a semantically secure public-key encryption scheme, except that the algorithm $G$ outputs $\ringid$ along with the public and private keys, and the set of values that can be encrypted using the public-key output by $G$ are the elements of $\ringfamily_\ringid$.  Furthermore, it is guaranteed that $|\ringfamily_\ringid| > 2^{k}$, and $|\ringfamily_\ringid| < 2^{qk}$ for some universal constant $q$.

\item Given $(pk,sk,\ringid) \leftarrow G(1^k)$, for any $x_1, x_2 \in \ringfamily_\ringid$, and given two
ciphertexts $c_1 = E(pk, x_1)$ and $c_2 = E(pk, x_2)$, we have that
$C(pk, c_1, c_2)$ outputs a distribution whose statistical distance
to the distribution $E(pk, x_1 + x_2)$ is negligible in $k$.

\item Given $(pk,sk,\ringid) \leftarrow G(1^k)$, for any $x, \alpha \in \ringfamily_\ringid$, given a ciphertext $c = E(pk, x)$, we have that $C(pk, c, \alpha)$ outputs a distribution whose statistical distance to the distribution $E(pk, \alpha \cdot x)$ is negligible in $k$.
\end{enumerate}

\end{definition}

In the case of uncontrolled ring homomorphic encryption schemes, we
will not consider general rings, but rather focus our attention on
the special case of $\mathbb{Z}_M$ (i.e. $\mathbb{Z}/M\mathbb{Z}$).
Here, we will assume that we are using the standard representation
of this ring (as integers in $[0,M-1]$ working modulo $M$).
We note that this is our only protocol where a specific representation of the underlying ring is important and required for our result.
In this
case, using a little bit of standard additional machinery, we can
once again construct a quite simple protocol for our basic two-party
functionality \fmult, as a show below.

\begin{itemize}
\item {\bf Protocol $\psi$.} \alice holds $a \in \mathbb{Z}_M$ and \bob holds $b \in \mathbb{Z}_M$.

    \begin{itemize}

    \item (Initialization)  Let $k' = \lceil 2 \log M \rceil + 2 + k$.  \alice runs $G(1^{k'})$ to obtain $(pk,sk,N)$, where $N > 4(2^k M^2)$.  This is done only once, as the same public key can be used as many times as necessary.

    \item \alice computes $c = E(pk, a)$, and sends $c$ to \bob.

    \item \bob chooses $r \in \mathbb{Z}_M$ and $s \in \mathbb{Z}_{2(2^k M)}$ at random, computes $c' = E(pk,r)$, $c'' = E(pk, sM)$, and then computes $c'''$ using the algorithm $C$ repeatedly so that $c'''$ is an encryption of $ab+r +sM$.  Note that $ab+r +sM < N$, by choice of parameters. \bob then sends $c'''$ to \alice, and outputs $-r \mod M$.

    \item \alice computes $v = D(sk,c''')$, and outputs $v \mod M$.
    \end{itemize}
\end{itemize}

A straightforward counting argument shows that for any $a,b,r \in
\mathbb{Z}_M$, setting $w = ab+r \mod M$, we have that the
statistical distance between the distributions $D_1 = (ab+r +sM)$
and $D_2 = (w+sM)$, where $s \in \mathbb{Z}_{2(2^k M)}$ is chosen at
random, is at most $2^{-k}$.  This is because $ab+r \le 2M^2$, and
so there are at most $2M$ choices of $s$ for which $w+sM$ would not
be in the support of $D_1$.  Thus, by the definition of uncontrolled
ring homomorphic encryption, the correctness and privacy properties
of this protocol (against passive corruptions) follow immediately.

\paragraph{Matrix rings.} Although we focus on the case of $\mathbb{Z}_M$ above, it is easy
to see that this approach can generalized to other related settings,
such as the ring of $n$ by $n$ matrices over $\mathbb{Z}_M$, in a
straightforward manner.  At a high level, this is because any
$\mathbb{Z}_M$-homomorphic encryption scheme immediately gives rise
to an encryption scheme that is homomorphic for the ring of $n$ by
$n$ matrices over $\mathbb{Z}_M$. In this context, by simply
encrypting each entry in the matrix, the homomorphic property of
matrix addition would follow immediately from the homomorphic
property with respect to addition of the underlying encryption
scheme.  The slightly interesting case is the ``scalar''
multiplication (by a known matrix) property of the homomorphic
encryption scheme.  It is easy to see that this property also holds,
since each entry of the product matrix is just a degree-2 function
of the entries of the two matrices being multiplied. Thus, for
instance in our case of $n$ by $n$ matrices, one can compute the
\fmult functionality with only $O(n^2)$ ciphertexts communicated,
even though no algebraic circuits for matrix multiplication are
known (or generally believed to exist) with $O(n^2)$ gates.

The discussion regarding matrices above is implicitly written in the
context of controlled-ring homomorphic encryption. In the context of
uncontrolled-ring homomorphic encryption, using the same ideas,
Protocol $\psi$ can be directly adapted to allow one to compute $m$
degree-2 functions over $n$ variables while communicating only
$O(m+n)$ ciphertexts. This allows one to use uncontrolled-ring
homomorphic encryption to compute the \fmult functionality for $n$
by $n$ matrices over $\mathbb{Z}_M$ with only $O(n^2)$ ciphertexts
(for an encryption scheme over $\mathbb{Z}_N$ where $\log N$ is $O(\log n+k+\log M)$) being communicated.

\section{General Arithmetic Computation against Active Corruption}
\label{sec:general}

As already discussed in Section~\ref{sec-techniques}, our general
protocols are obtained by applying the general technique of~\cite{IshaiPrSa08}, with appropriate choices of the ``outer protocol'' and the ``inner protocol'' that apply to the arithmetic setting.

More concretely, the result from~\cite{IshaiPrSa08} shows how obtain a UC-secure protocol in the
\OT-hybrid model for any (probabilistic polynomial time) two-party functionality $f$
against active corruption by making a {\em black-box} use of the following two ingredients:
\begin{enumerate}
\item an ``outer protocol'' for $f$ which employs $k$ auxiliary parties (servers); this protocol should be UC-secure against active corruption provided that only some constant fraction the servers can be corrupted; and
\item an ``inner protocol'' for implementing a reactive two-party functionality (``inner functionality'') corresponding to the local computation of each
server, in which the server's state is secret-shared between Alice and Bob. In contrast to the outer protocol, this protocol only needs to be secure against {\em passive} corruption.
The inner protocol can be implemented in the OT-hybrid model.
\end{enumerate}

While the general result of~\cite{IshaiPrSa08} is not sensitive to the type
of secret sharing used for defining the inner functionality, in our setting
it is crucial that any ring elements stored by a server will be
secret-shared between Alice and Bob using additive secret sharing over the
ring. Given our protocols for \fmult, this will let us have the
the inner protocol use the ring in a black-box fashion, as described below.

Note that the only operations that the server in an outer protocol needs to
do  one of the following operations: add two ring elements, multiply two
ring elements, sample a ring element uniformly at random, or check if two
ring elements are equal.  If there are oprations which do not involve any
ring elements, the inputs and outputs to these operations are maintained as
bit strings and an arbitrary protocol for boolean circuit evaluation (e.g.,
GMW in the \OT-hybrid model) can be employed. Among the operations that do
involve ring elements, addition and sampling are straightforward: whenever a
server in the outer protocol needs to locally add two ring elements $x,y$,
this can be done locally in the inner protocol by having each of Alice and
Bob add their local shares of the two secrets.  When a server in the outer
protocol needs to sample a random ring element, Alice and Bob locally sample
the shares of this element. For multiplication, when a server needs to
multiply two ring elements $x,y$ in the outer protocol, the inner protocol
will need to apply a sub-protocol for the
following two-party functionality:
\begin{itemize}
\item \alice holds $x_\alice$ and $y_\alice$, \bob holds $x_\bob$ and $y_\bob$.
\item The server should compute random values $c_\alice$ and $c_\bob$ such that
$c_\alice+c_\bob = (x_\alice+x_\bob)(y_\alice+y_\bob)$.
\item \alice is given $c_\alice$ and \bob is given $c_\bob$.
\end{itemize}

The above functionality can be realized (in the semi-honest model) by making
two calls to any of the product-sharing protocols from \sectionref{passive}.
Specifically, a secure reduction from the above functionality to \fmult may
proceed as follows:
\begin{itemize}
\item \alice and \bob engage in two instances of \fmult with
inputs $(x_\alice,y_\bob)$ and $(y_\alice,x_\bob)$ and obtain
$(\alpha_\alice,\alpha_\bob)$ and $(\beta_\alice,\beta_\bob)$ where
$\alpha_\alice+\alpha_\bob = x_\alice y_\bob$ and
$\beta_\alice+\beta_\bob = y_\alice x_\bob$.
\item \alice outputs $c_\alice := x_\alice y_\alice + \alpha_\alice +
\beta_\alice$ and \bob outputs
$c_\bob := x_\bob y_\bob + \alpha_\bob + \beta_\bob$.
\end{itemize}

There will be several such instances of \fmult in each round.  Note that  Protocol~\pmulttx
can be used to realize multiple instances of
\fmult with a constant amortized algebraic complexity per instance.

The final type of computation performed by servers involving ring elements
is equality check between two ring elements. In all the outer protocols we
employ, the result of such an equality test is made public. (In fact, in our
setting of ``security with abort,'' the outer protocols we consider will
abort whenever an inequality is detected by an honest server in such an
equality test.) The corresponding inner functionality needs to check that
$x_a+x_b=y_a+y_b$, where $x_a,y_a$ are identifiers of ring elements known to
Alice and $x_b,y_b$ are known to Bob. One way to do this would be by letting
Alice locally compute $x_a-y_a$, Bob locally compute $y_b-x_b$, and then
using an arbitrary inner protocol for boolean circuits for comparing the two
identifiers. This relies on our assumption that each ring element has a
unique identifier.  However, in fact in the outer protocols we consider,
there is a further structure that allows us to avoid this generic approach.
The elements to be compared by a server in our outer protocols will always
be known to one of the parties (Alice or Bob), and hence in a passive-secure
implementation this comparison can be done locally by that party. (This is
referred to as a ``type I computation'' in~\cite{IshaiPrSa08}. Note that
given a passive-secure implementation, the compiler of \cite{IshaiPrSa08}
ensures over all security.)

Below we summarize the results we obtain by combining appropriate choices for the outer protocol with the inner protocols obtained via the shared-product protocols from
\sectionref{passive}.
All these results
can be readily extended to the multi-party setting as well, where the complexity grows polynomially with the number of parties; see Appendix~\ref{app-multiparty}.

\paragraph{Unconditionally secure protocol.}
To obtain our unconditional feasibility result for black-box rings, we use the protocol from~\cite{CramerFIK03} (which makes a black-box use of an arbitrary ring) as the outer protocol and the unconditional protocol \pmultbasicx to build the inner protocol. This yields the following result:
\begin{theorem}
For any arithmetic circuit \ckt, there exists a protocol $\Pi$ in the
\OT-hybrid model that is a secure black-box realization of \ckt-evaluation
for the set of all ring families. The security holds against adaptive
corruption with erasures, in computationally unbounded environments.
\end{theorem}

The {\em arithmetic} communication complexity of the protocol \pmultbasicx,
and hence that of the above protocol,
grows linearly with (a bound on) $|\log\ringfamily_\ringid|$. (Recall that, by convention, the required upper bound is given by $|\ringid|$; otherwise such a bound can be inferred from the length of identifiers.)

\paragraph{Protocols from noisy encodings.}
To obtain a computationally secure protocol
whose arithmetic communication complexity is independent of the ring, we
shall depend on Assumption~\ref{asm:lincode-generic}, instantiated with the
code generation algorithm \randcodegen based on random linear codes.  By
replacing \pmultbasicx by \pmultlcx (with \randcodegen as the code
generation scheme) in the previous construction we obtain the following:
\begin{theorem}
Suppose that Assumption~\ref{asm:lincode-specific}(a) holds. Then, for every arithmetic circuit
\ckt, there exists a protocol $\Pi$ in the \OT-hybrid model that
is a secure black-box realization of  \ckt-evaluation
for the set of all computationally efficient field families \fieldfamily.
The security holds against adaptive
corruption with erasures.  Further, the arithmetic complexity of $\Pi$ is $\poly(k)\cdot|\ckt|$,
independently of \fieldfamily or $\ringid$.
\end{theorem}

Using \ringcodegen instead of \randcodegen, this result extends to all
ring families for which  Assumption~\ref{asm:lincode-generic} holds
with \ringcodegen. Recall that we propose this assumption
for {\em all efficient} computational ring families \ringfamily.
\begin{theorem}
Suppose that Assumption~\ref{asm:lincode-specific}(b) holds. Then, for every arithmetic circuit
\ckt, there exists a protocol $\Pi$ in the \OT-hybrid model that
is a secure black-box realization of  \ckt-evaluation
for the set of all computationally efficient ring families \ringfamily.
The security holds against adaptive
corruption, with erasures.  Further, the arithmetic complexity of $\Pi$ is $\poly(k)\cdot|\ckt|$,
independently of \ringfamily or $\ringid$.
\end{theorem}

Finally, our most efficient protocol will be obtained by using a variant of the protocol from \cite{DamgardIs06} as the outer
protocol (see Appendix~\ref{app-outer}) and an inner protocol which is based on \pmulttx
(with $n=O(k)$ and $t=\Omega(k)$). To get the specified computational complexity, the size of the field should be super-polynomial in the security parameter. (The communication complexity does not depend on this assumption.)
\begin{theorem}
Suppose that Assumption~\ref{asm:lincode-specific}(c) holds. Then, for every arithmetic circuit
\ckt, there exists a protocol $\Pi$ in the \OT-hybrid model
with the following properties.
The protocol $\Pi$ is a secure black-box realization of  \ckt-evaluation
for the set of all computationally efficient field families \fieldfamily, with respect to all computationally bounded environments for which $|\fieldfamily_\ringid|$ is super-polynomial in $k$.
The security of $\Pi$ holds against adaptive
corruption with erasures. The arithmetic communication complexity of $\Pi$ is
$O(|\ckt| + k\cdot {\mathsf{depth}}(\ckt))$,
where
${\mathsf{depth}}(\ckt)$ denotes the depth of \ckt, and its arithmetic computation complexity is $O(\log^2k)\cdot(|\ckt|+ k\cdot  \mathsf{depth}(\ckt))$. Its round complexity is $O(\mathsf{depth}(\ckt))$.
\end{theorem}
By using a suitable choice of fields and evaluation points for the Reed-Solomon encoding (see Section~\ref{sec:basic-packed}), and under a corresponding specialization of Assumption~\ref{asm:lincode-specific}(c), the computational overhead of the above protocol can be reduced from $O(\log^2k)$ to $O(\log k)$. (In this variant we do not attempt to make a black-box use of the underlying field and rely on the standard representation of field elements.)

\paragraph{Protocols from homomorphic encryption.}
We also consider protocols which make a black-box\footnote{Here and in the following, when saying that a construction makes a black-box use of a homomorphic encryption primitive we refer to the notion of a fully black-box reduction, as defined in~\cite{ReingoldTrVa04}. This roughly means that not only does the construction make a black-box use of the primitive, but also its security is proved via a black-box reduction.
}
 use of homomorphic encryption. These are obtained in a manner similar to above, but using protocols $\theta$ and $\psi$ as the inner protocols and~\cite{CramerFIK03} as the outer protocol. Using these we obtain the following theorems:
\begin{theorem}
For every arithmetic circuit
\ckt, there exists a protocol $\Pi$ in the \OT-hybrid model, such that for every ring family $\ringfamily$, the protocol $\Pi^\ringfamily$ securely realizes ${\cal F}_C^\ringfamily$ by making a {\em black-box} use of any controlled-ring homomorphic encryption for \ringfamily.
The security holds against adaptive
corruption with erasures.
The number of invocations of the encryption scheme is $\poly(k)\cdot|\ckt|$, independently of \ringfamily or $\ringid$.
\end{theorem}
Note that the above theorem can be instantiated with the ring of $n$ by $n$ matrices over $\mathbb{Z}_p$,
and the communication complexity of the resulting protocol would be $\poly(k)\cdot|\ckt|\cdot n^2$. Combined with~\cite{MohasselW08}, this yields constant-round protocols for secure linear algebra which make a black-box use of homomorphic encryption and whose communication complexity is nearly linear in the input size.

For the case of fields, we obtain the following more efficient version of the result by using the efficient outer protocol from Appendix~\ref{app-outer}:
\begin{theorem}
For every arithmetic circuit
\ckt, there exists a protocol $\Pi$ in the \OT-hybrid model, such that for every field family $\fieldfamily$, the protocol $\Pi^\fieldfamily$ securely realizes ${\cal F}_C^\ringfamily$ by making a {\em black-box} use of any controlled-ring homomorphic encryption for \fieldfamily.
The security holds against adaptive
corruption with erasures.
Further, $\Pi$ makes $O(|\ckt| + k\cdot {\mathsf{depth}}(\ckt))$ invocations of the encryption scheme, and the communication complexity is dominated by sending $O(|\ckt| + k\cdot {\mathsf{depth}}(\ckt))$ ciphertexts.
\end{theorem}

We also obtain analogous results for uncontrolled-ring homomorphic encryption:
\begin{theorem}
For every arithmetic circuit \ckt
there exists a black-box construction of a protocol $\Pi$ in the \OT-hybrid model from any uncontrolled-ring homomorphic encryption for the standard representation of the ring family $\mathbb{Z}_{M}$, such that $\Pi$ is a secure realization of  \ckt-evaluation for the same ring family under the standard representation.
The security holds against adaptive
corruption with erasures.
The number of invocations of the encryption scheme is $\poly(k)\cdot|\ckt|$, independently of $\ringid$, and the communication complexity is dominated by $\poly(k)\cdot|\ckt|$ ciphertexts.  During the protocol, the ring size parameter fed to the encryption scheme by honest parties is limited to $k'=O(k+|\ringid|)$.

If, further, the ring over which $C$ should be computed is restricted to be a field, there exists a protocol as above which makes $O(|\ckt| + k\cdot {\mathsf{depth}}(\ckt))$ invocations of the encryption scheme, and where the communication complexity is dominated by sending $O(|\ckt| + k\cdot {\mathsf{depth}}(\ckt))$ ciphertexts.
\end{theorem}
The efficient version of the above theorem also applies to the case of arithmetic computation over pseudo-fields, in scenarios where it is computationally hard to find zero divisors. Furthermore, it can be generalized to the ring of $n$ by $n$ matrices, which when used with constructions of uncontrolled-ring $\mathbb{Z}_N$-homomorphic encryption schemes from the literature~\cite{Paillier99,DamgardJ02} would yield arithmetic protocols for matrices over large rings whose complexity grows quadratically with $n$.

We finally note that in the {\em stand-alone} model, the OT oracle in the above protocols can be realized by making a black-box use of the homomorphic encryption primitive without affecting the asymptotic number of calls to the primitive. This relies on the black-box construction from~\cite{IshaiKLP06} and the fact that only $O(k)$ OTs need to be secure against active corruption. 
Thus, the above theorems hold also in the plain, stand-alone model (as opposed to the OT-hybrid UC-model), assuming that the underlying ring has identity.\footnote{The identity element is used in the standard construction of semi-honest OT from homomorphic encryption.}

\superparagraphb{Acknowledgments.} We thank Jens Groth, Farzad Parvaresh, Oded Regev, and Ronny Roth for helpful discussions.

\bibliographystyle{alpha}
\bibliography{bib}

\newpage
\appendix
\section{Security Against Adaptive Passive Corruption, with Erasures}
\label{app:adaptive}
Here we present a general transformation, that
applies to a class of protocols covering all our protocols from Section~\ref{sec:passive}, to obtain
protocols that are secure against \emph{adaptive} passive corruption, in the model with erasures.  This transformation is done in a simple way using standard techniques.  At a high level, the idea is simply to call our basic protocol on randomly chosen inputs, erase the ``local computations'' done while executing the basic protocol, and then communicate ``corrections'' in order to convert the outputs of the random execution into the desired outputs for the real inputs (in a manner very similar to Beaver's reduction of OT to random OT~\cite{Beaver95}).  Intuitively, in our new protocol, if an adversary adaptively corrupts a party during the initial random invocation of the basic protocol, there is no problem since the protocol was anyway run on random inputs chosen independently of the parties' actual inputs (although this is not quite accurate, which is why we introduce a notion of ``special simulation'' below).  On the other hand, if the adversary corrupts a party after the basic protocol is done, then since the party has already erased the local computations of the protocol, we are free to choose a ``random-looking'' output from the basic protocol in such a way that we can use it to explain the actual inputs and outputs that we have.

The main protocol in this section, \pmult has security against passive adaptive corruption, with erasures. \pmult is built using any
protocol \pmultx with a simpler security property described below.  Applying the transformation in this section also has another efficiency advantage in scenarios where pre-processing interaction is possible, and this is discussed briefly in a remark at the end of this section.

\subsection{Special Simulation Security Against Passive Corruption}
It will be convenient for us to introduce an intermediate notion of
security of multi-party computation against {\em static} passive
corruption, which will then enable us to obtain security against
{\em adaptive} passive corruption with erasures. This intermediate
security property is quite weak, and is required to hold only
against random inputs (though the candidates we shall use later in
fact satisfy stronger security).

Let \Ffunc be a secure function evaluation
functionality. We use the following terminology.
\begin{itemize}
\item
An environment \env is said to be a {\em random-input environment} if it provides
independent random inputs (according to a specified distribution) to each
party.
\item
A simulator \simx is said to be a {\em special simulator} if it behaves as
follows:
\begin{enumerate}
\item \simx sends the corrupt parties' inputs to \Ffunc, and obtains the
outputs from \Ffunc.
\item \simx picks {\em random inputs for all the honest parties to be
simulated}. \simx also sets the random tapes of all the parties (corrupt
and honest). These choices are (jointly) indistinguishable from the uniform (or
specified) distribution, even given the input of the corrupt parties.
(However \simx can correlate these choices with the output obtained from
\Ffunc).
\item \simx ensures that on interacting with the simulated honest parties,
the corrupt parties will produce the {\em same outputs as given by \Ffunc}.
If this is not the case \simx will abort. Otherwise, it
reports the view of the corrupt parties in this execution to the
environment.
\end{enumerate}
\end{itemize}

\begin{definition}
\label{def:splsimsec}
A protocol \prot is said to {\em securely realize \Ffunc
on random inputs, against passive corruption, with special simulation}
if there exists a special simulator \simx such that for all random-input
environments \env and a static passive adversary \adv,
the real execution of the protocol \prot between the
parties is indistinguishable to \env, from an ideal execution of
the parties interacting with \Ffunc and \simx.
\end{definition}

\subsection{Special Simulation Security to Security Against Adaptive Corruption with Erasures}

Given a protocol \pmultx which securely realizes \fmult
on random inputs, against passive corruption, with special
simulation, below we show how to construct a protocol \pmult
with security against adaptive passive corruption, with erasures.

\begin{itemize}
\item {\bf Protocol \pmult.} \alice holds $a \in \ring$ and \bob holds
$b \in \ring$.
    \begin{enumerate}
    \item \alice picks $r^\alice\in\ring$ and \bob picks $r^\bob\in\ring$ at random.
    \item \alice and \bob run \pmultx with inputs $r^\alice$ and $r^\bob$
        respectively, and obtains outputs $s^\alice$ and $s^\bob$ respectively. Note
        that $s^\alice+s^\bob = r^\alice r^\bob$.
    \item \alice and \bob erase the memory used for \pmultx. (They retain the
        inputs and outputs, namely $(r^\alice,s^\alice)$ and $(r^\bob,s^\bob)$ respectively.)
    \item \alice sends $a-r^\alice$ to \bob, and \bob sends $b-r^\bob$ to \alice.
    \item \alice outputs $z^\alice := a(b-r^\bob) + s^\alice$; \bob outputs $z^\bob :=
        (a-r^\alice)r^\bob + s^\bob$. Note that $z^\alice+z^\bob = ab$.
    \end{enumerate}
\end{itemize}

We now show the following:
\begin{lemma}
\label{lem:adaptive} For any \pmultx which securely realizes \fmult
on random inputs, against passive corruption, with special
simulation, protocol \pmult is a secure realization of \fmult
against adaptive passive corruption with erasures. If the security
of \pmultx is statistical, so is that of \pmult.
\end{lemma}

\begin{sketch}
The interesting cases are when during the protocol initially \alice is corrupted
and later \bob is corrupted, or when initially \bob is corrupted and later
\alice is corrupted. (Recall that all corruptions are passive.)

Let \envad be an arbitrary environment which gives \alice and \bob inputs for
\pmult.
We will consider the case when \alice is corrupted initially and \bob may be
corrupted later. The other case is symmetric for this analysis.
Our
simulator \simu works as follows.

\begin{itemize}
\item \simu sends $a$ to \fmult and obtains $z^\alice$ from it.
\item \simu picks a random value $c\from\ring$ and sets $s^\alice :=
z^\alice - ac$, and also picks a random value $r^\alice\from\ring$.

\item Next \simu internally runs the special simulator \simx for \pmultx with
$r^\alice$ as input to \alice. \simx expects to interact with an instance of
\fmult, which, for clarity, we will denote by \fmultx.  \simx simulates \fmultx
by providing $s^\alice$ as the output for \alice.

\item If \envad instructs to corrupt \bob before this simulation finishes and the
erasure step (step 3) is simulated, then \simu obtains the simulated state
for \bob in \pmultx from \simx; then \simu constructs a state for \bob in \pmult by
combining this with $b$, the input to \bob (which is not used until step 4 of
the protocol), and reports this to \envad.

\item If \bob is still not corrupted at step 4, \simu uses the value $c$
as the simulated message from \bob to \alice in step 4. Note that $z^\alice
= ac+s^\alice$.

\item By this step \bob has already erased his state during the execution of
\pmultx. So if \bob is corrupted at any point after this, its state can be
explained by giving $(b,r^\bob,s^\bob)$: for this the simulator \simu will
obtain $b$ by corrupting \bob in the ideal world, and set $r^\bob := b-c$
and $s^\bob := r^\alice r^\bob - s^\alice$, where $r^\alice$ and $s^\alice$
are the input and output of \alice in the simulated execution of \pmultx.
Note that this pair $(r^\bob,s^\bob)$
is consistent with $(r^\alice,s^\alice)$ by the functionality \fmult.
\end{itemize}

The indistinguishability of simulation follows from the two
requirements on the simulator \simx for \pmultx: that it is a special simulator and that it
provides an indistinguishable simulation against static corruption (on
random inputs).

If \bob is corrupted before step~4, the simulated execution is
indistinguishable from the real execution.
To see this, firstly note that
\simx is given $s^\alice$ as the output from \fmultx, but this is
indeed a random element (because $z^\alice$ is random).
Then, \simx is guaranteed to set the random tape of \alice, as well as
\bob's input and random tape to be indistinguishable from uniformly random
choices.
So the simulated state of \alice and \bob are indistinguishable from the
real execution up to step~4. The simulated state of \bob is completed by
incorporating \bob's input $b$ (the state used in execution before step~4
being independent of $b$).

If \bob is corrupted after step~4, then we consider the following two
experiments, with an environment \envx which consists of the given environment
\envad as well as part of our simulator which picks $c$ (but does not get
$z^\alice$ or compute $s^\alice$).
The environment \envx provides $r^\alice$ as input to \alice and $r^\bob := b-c$ to \bob.
It outputs the bit output by \envad.
\begin{itemize}
\item[\real:]
\alice and \bob execute \pmultx on their inputs from \envx.
\item[\ideal:]
In the \ideal execution \fmultx gives a random pair $(s^\alice,s^\bob)$ such
that $s^\alice+s^\bob = r^\alice r^\bob$. \simx interacts with \envx
simulating the internal state of \alice.
\end{itemize}
By the security requirement on \simx, the two experiments are indistinguishable
to the environment \envx.
Further the \real experiment above is identical to the \real execution of \pmult
with \envad. To complete the proof we need to argue that the \ideal execution above (with \envx) is identical to
our \ideal execution (with \envad).
Note that in our description of
the simulation
$s^\alice := z^\alice-ac$, whereas in the \ideal execution with \envx,
$s^\alice$ is just
a random element.
However, though the environment \envx knows $a$ and $c$,
$z^\alice$ will be picked at random (by \fmult in our \ideal execution).
In other words, we could consider a modified \fmultx which receives $ac$ from the
environment \envx, then picks a random element $z^\alice$ and sets $s^\alice := z^\alice -ac$,
without altering the experiment. With this modification, our \ideal
execution (with \envad and \simu) is identical to the \ideal execution with \envx and
\simx.

\end{sketch}

In \sectionref{passive}, for the protocols \pmultx, \pmultlcx and \pmulttx
we showed security against static passive corruption (even for non-random
inputs). The simulators we used in these proofs are in fact special
simulators.  The same is easily seen to be true for protocols $\theta$ and
$\psi$ based on homomorphic encryption. Thus we have the following result.

\begin{lemma}
Protocols \pmultx, \pmultlcx, \pmulttx, $\theta$, and $\psi$
securely realize $\fmult$ (or \fmultt in the case of Protocol
\pmulttx), on random inputs, against passive corruption, with
special simulation.
\end{lemma}

Hence, by plugging them into the protocol \pmult we obtain
corresponding protocols which are secure against passive, adaptive
corruption with erasures.

\paragraph{Remark.} The structure of the protocol in this section allows
an efficiency gain by employing pre-processing. The first half of the
protocol, which is executed on random inputs can be carried out before the
actual function evaluation starts. Further, when used to implement a
reactive functionality, the entire set of steps  involving \OT that will be
ever used in the lifetime of the protocol can be carried out up front. In
fact, later in applying our protocols as the ``inner protocol'' of the final
construction, we can use this reactive variant.

\section{Extension to Multi-Party Computation}
\label{app-multiparty}

In this section, we briefly sketch what is involved in extending our results to the multi-party case.

The protocol in \cite{IshaiPrSa08} extends to more than two parties,
given inner and outer protocols for that many parties. The outer protocols
from \cite{DamgardIs06} and \cite{CramerFIK03} do extend to the multi-party
setting (called the ``multi-client'' setting in~\cite{IshaiPrSa08} for more details).  Hence by extending our inner protocol to the multi-party
setting, all our results extend similarly.

In the general multi-party case the only non-trivial kind of computations carried out by the servers
in the outer protocol is as follows:
\begin{itemize}
\item Each party $P_i$ ($i=1,\ldots,m$) sends $x_i$ and $y_i$ to the server.
\item The server computes random values $c_i$ such that $\Sigma_i c_i = (\Sigma_i x_i)(\Sigma_i y_i)$. Each party $P_i$ is given $c_i$ as the output.
\end{itemize}

A protocol for this using \fmult is as follows:
\begin{itemize}
\item For each ordered pair $(i,j)$, $i\not=j$, parties $P_i$ and $P_j$ engage in an instance of \fmult with inputs $(x_i,y_j)$ and obtain outputs $(\alpha^{(i,j)}_i,\alpha^{(i,j)}_j)$, respectively, where $\alpha^{(i,j)}_i+\alpha^{(i,j)}_j = x_i y_j$.
\item $P_i$ outputs $x_iy_i + \Sigma_{j\not=i} \left( \alpha^{(i,j)}_i + \alpha^{(j,i)}_i \right)$.
\end{itemize}

The correctness of this protocol, and its (perfect) privacy against passive corruptions, is standard and analogous to the binary case from~\cite{GoldreichMiWi87}.

\section{An Efficient Outer MPC Protocol}
\label{app-outer}

In this section, which is adapted from a preliminary full version of~\cite{IshaiPrSa08}, we describe a variant of the protocol from~\cite{DamgardIs06} which we use as the efficient
outer MPC protocol in our constructions. We restrict the attention to the case of black-box fields (alternatively, pseudo-fields), and assume that the field size is super-polynomial in the security parameter. (This assumption can be removed at a minor cost to the arithmetic complexity.)

The protocol involves $n$ servers
and $m$ clients ($m=2$ by default),
where only clients have inputs and outputs.
The protocol is statistically UC-secure against an adaptive adversary corrupting an arbitrary number of clients and some constant fraction of the servers. We note that unlike the protocol from~\cite{DamgardIs06}, here we do not need to guarantee output delivery and may settle for the weaker notion of ``security with abort''. This makes the protocol simpler, as it effectively means that whenever an inconsistency is detected by an honest party, this party can broadcast a complaint which makes the protocol abort.

For simplicity we assume that all $n+m$ parties in the MPC protocol have common access to an oracle which broadcasts random field elements, and do not count these elements towards the communication complexity. In~\cite{DamgardIs06} this is emulated via a distributed coin-flipping protocol and an $\epsilon$-biased generator~\cite{NaorNa90}, which reduce the communication cost of implementing this procedure. Alternatively, random field elements can be directly generated by the $m$ clients in the final protocol via efficient coin-flipping in the OT-hybrid model.

Before describing the protocol, we summarize its main efficiency features. For simplicity we shall restrict ourselves to $n=O(k)$, where k is a statistical security parameter, and a constant number of clients $m$. To
evaluate an arithmetic circuit $C$ of size $s$ and multiplicative depth $d$, the arithmetic communication complexity is $O(s+kd)$.\footnote{While we do not attempt here to optimize the additive $O(kd)$ term, we note that a careful implementation of the protocol seems to make this term small enough for practical purposes. In particular, the dependence of this term on $d$ can be eliminated for typical circuits.}
Assuming broadcast as an atomic primitive, the protocol requires $O(d)$
rounds. (We note that in the final $m$-party protocol obtained via the technique of~\cite{IshaiPrSa08}, broadcast only needs to be performed among the clients; in particular, in the two-party case broadcast can be implemented by directly sending the message.)

The computational complexity will be addressed after we describe the protocol.

To simplify the following exposition we will only consider the case
of two clients Alice and Bob. An extension to the case of a larger
number of clients is straightforward.

Another simplifying assumption is that the circuit $C$ consists of
$d$ layers, where each layer performs addition, subtraction, or multiplication operations
on values produced by the previous layer only. Circuits of an
arbitrary structure can be easily handled at a worst-case additive
cost of $O(nd)$, independently of the circuit size.
(This cost can be amortized away for almost
any natural instance of a big circuit. For instance, a sufficient
condition for eliminating this cost is that for any two connected
layers there are at least $n$ wires connecting between the layers.)

\subsection{Building Blocks}

The protocol relies on tools and sub-protocols
that we describe below.

\paragraph{Secret sharing for blocks.}
Shamir's secret sharing scheme~\cite{Shamir79} distributes a secret
$s\in F$ by picking a random degree-$\delta$ polynomial $p$ such that
$p(0)=s$, and sending to server $j$ the point $p(j)$. Here $F$ is a
finite field such that $|F|>n$. By $1,2,\ldots,n$ we denote distinct interpolation points, which in the case of a black-box access to $F$ can be picked at random. The generalization of Franklin and
Yung~\cite{FranklinYu92} achieves far better efficiency with a minor cost to the security level. In this scheme, a {\em block} of $\ell$ secrets
$(s_1,\ldots,s_\ell)$ is shared by picking a random degree-$\delta$
polynomial $p$ such that $p(1-j)=s_j$ for all $j$, and distributing
to server $j$ the point $p(j)$. (Here we assume that
$-\ell+1,\ldots,n$ denote $n+\ell$ distinct field elements.) Any set
of $\delta+1$ servers can recover the entire block of secrets by
interpolation. On the other hand, any set of $t=\delta-\ell+1$ servers
learn nothing about the block of secrets from their shares. (Secret
sharing schemes in which there is a gap between the privacy and
reconstruction thresholds are often referred to as ``ramp
schemes''.) For our purposes, we will choose $\ell$ to be a small
constant fraction of $n$ and $\delta$ a slightly bigger constant fraction
of $n$ (for instance, one can choose $\delta=n/3$ and $\ell=n/4$). This
makes the amortized communication overhead of distributing a field element constant, while maintaining secrecy against a constant fraction of the servers.

\paragraph{Adding and multiplying blocks.} Addition (or subtraction) and multiplication of shared blocks is analogous to the use of Shamir's scheme in the BGW protocol~\cite{BenOrGoWi88}. Suppose that a block $a=(a_1,\ldots,a_\ell)$
was shared via a polynomial $p_a$ and a block
$b=(b_1,\ldots,b_\ell)$ was shared via a polynomial $p_b$. The
servers can then locally compute shares of the polynomial $p_a+p_b$,
which are valid shares for the sum $a+b$ of the two blocks. If each
server multiplies its two local shares, the resulting $n$ points are
a valid secret-sharing using the degree-$(2\delta)$ polynomial $p=p_ap_b$
of the block $ab=(a_1b_1,\ldots,a_\ell b_\ell)$. Note, however, that
even if $p_a,p_b$ were obtained from a random secret sharing,
$p_ap_b$ is not a random degree-$(2\delta)$ secret sharing of $ab$. Thus,
if we want to reveal $ab$ we will need to mask $p_ap_b$ by a random
degree-$2d$ secret-sharing of a block of 0's before revealing it.
Also, in order to use $ab$ for subsequent computations we will need
to reduce its degree back to $\delta$.

\paragraph{Proving membership in a linear space.} The protocol will often
require a client to distribute to the servers a vector $v=(v_1,\ldots,v_n)$
(where each $v_j$ includes one or more field elements) while assuring them
that $v$ belongs to some linear space $L$. This should be done while
ensuring that the adversary does not learn more information about $v$ than
it is entitled to, and while ensuring the honest parties that the shares
they end up with are consistent with $L$. For efficiency reasons, we settle
for having the shares of the honest parties {\em close} to being consistent
with $L$. Since we will only use this procedure with $L$ that form an error
correcting code whose minimal distance is a large constant multiple of $\delta$,
the effect of few  ``incorrect'' shares can be undone via error-correction.
(In fact, in our setting of security with abort error detection will be
sufficient.) More concretely, our procedure takes input
$v=(v_1,\ldots,v_n)\in L$ from a dealer $D$ (Alice or Bob). In the presence
of an active, adaptive adversary who may corrupt any client and at most $t$
servers, it should have the following properties:

\begin{itemize}
\item
Completeness: If $D$ is uncorrupted then every honest server $j$ outputs
$v_j$.

\item Soundness: Suppose $D$ is corrupted. Except with negligible
probability, either all honest servers reject (in which case the dealer is
identified as being a cheater), or alternatively the joint outputs of all
$n$ servers are most $2t$-far (in Hamming distance) from some vector in
$v\in L$.

\item Zero-Knowledge: If $D$ is uncorrupted, the adversary's view can be
simulated from the shares $v_j$ of corrupted servers.
\end{itemize}

Verifiable Secret Sharing (VSS) can be obtained by applying the
above procedure on the linear space defined by the valid share
vectors. Note that in contrast to standard VSS, we tolerate some
inconsistencies to the shares on honest servers. Such inconsistencies will be handled by the robustness of the higher level protocol.

\paragraph{Implementing proofs of membership.}
We will employ a sub-protocol from~\cite{DamgardIs06} (Protocol~5.1) for
implementing the above primitive. This protocol amortizes the cost
of proving that many vectors $v^1,\ldots, v^q$ owned by the same
dealer $D$ belong to the same linear space $L$ by taking
random linear combinations of these
vectors together with random vectors from $L$ that are used for
blinding.
The high level structure of this protocol is as follows.
\begin{itemize}
  \item {\em Distributing shares.} $D$ distributes $v^1,\ldots,v^q$ to the servers.

  \item {\em Distributing blinding vectors.} $D$ distributes a random vector
  $r\in L$ that is used for blinding. (This step ensures
  the zero-knowledge property; soundness does not depend on the valid choice
  of this $r$.)

  \item {\em Coin-flipping.}
The players invoke the random field element oracle to obtain a length-$q$ vector defining a random linear combination of the $q$ vectors distributed by the dealer.  (In~\cite{DamgardIs06} this is implemented using distributed coin-flipping and an $\epsilon$-biased generator; in our setting this can be implemented directly by the clients in the OT-hybrid model. Moreover, in the case of two
  clients we let the other client, who does not serve as a dealer, pick $r$ on its own.)

  \item {\em Proving.} The dealer computes the linear combination
  of its vectors $v^i$ defined by $r$, and adds to it the corresponding blinding vector. It broadcasts the results.

  \item {\em Complaining.} Each server applies the linear
combination  specified by $r$  to its part of the vectors distributed by the dealer, and
  ensures that the result is consistent with the value broadcast in the
  previous step. If any inconsistency is detected, the server broadcasts a complaint and the protocol aborts.
Also, the protocol aborts if the vector broadcasted by the dealer is not in $L$.

\item {\em Outputs.} If no server broadcasted a complaint, the servers output the shares distributed by
  the dealer in the first step (discarding the blinding vectors and the
  results of the coin-flips).
  \end{itemize}
In the case of a static corruption of servers, if the shares dealt to honest servers are inconsistent, the protocol will abort except with $1/|F|$ probability, which is assumed to be negligible in $k$. The adaptive case is a bit more involved, since the adversary can choose which servers to corrupt only after the random linear combination is revealed. This case is easy to analyze via a union bound (which requires that $|F|>{n\choose t}$). Alternatively, a tighter analysis shows that if $|F|$ is superpolynomial in $k$ then, except with negligible probability, either the protocol aborts or there exists a small set $B$ of servers such that all shares held by the honest servers {\em excluding} those in $B$ are consistent with a valid codeword from $L$. This condition is sufficient for the security of the protocol.

We will sometimes employ the above protocol in a scenario where
vectors $v^1,\ldots,v^q$ are already distributed between the servers
and known to the dealer, and the dealer wants to convince the
servers that these shares are consistent with $L$. In
such cases we will employ the above sub-protocol without the first
step.

\paragraph{Proving global linear constraints.} We will often need to deal
with a more general situation of proving that vectors $v^1,\ldots,v^q$ not
only lie in the same space $L$, but also satisfy additional global
constraints. A typical scenario applies to the case where the $v^i$ are
shared blocks defined by degree-$\delta$ polynomials. In such a case, we will
need to prove that the secrets shared in these blocks satisfy a specified
replication pattern (dictated by the structure of the circuit $C$ we want to
compute). Such a replication pattern specifies which entries in the $q$
blocks should be equal. An observation made in~\cite{DamgardIs06} is
that: (1) such a global constraint can be broken into at most $q\ell$ atomic
conditions of the type ``entry $i$ in block $j$ should be equal to entry
$i'$ in block $j'$'', and (2) by grouping these atomic conditions into
$\ell^2$ types defined by $(i,i')$, we can apply the previous verification
procedure to simultaneously verify all conditions in the same type. That is,
to verify all conditions of type $(i,i')$ each server concatenates his two
shares of every pair of blocks that should be compared in this type, and
then applies the previous verification procedure with $L$ being the linear
space of points on degree-$\delta$ polynomials $(p_1,p_2)$ which satisfy the
constraint $p_1(1-i)=p_2(1-i')$.
Unlike~\cite{DamgardIs06} we will also employ the above procedure in the
case where $p_1,p_2$ may be polynomials of different degrees (e.g., $\delta$ and
$2\delta$), but the same technique applies to this more general case as well.

\subsection{The Protocol}

The protocol is a natural extension of the protocol from~\cite{DamgardIs06},
which can be viewed as handling the special case of constant-depth circuits using a
constant number of rounds. We handle circuits of depth $d$ by using $O(d)$
rounds of interaction. The protocol from~\cite{DamgardIs06} handles general
functions by first encoding them into $\NCz$ functions, but such an encoding step is too expensive
for our purposes and in any case does not apply to the arithmetic setting. The protocol is simplified by the fact that we only need to achieve ``security with abort'', as opposed to the full security of the protocol from~\cite{DamgardIs06}.

Recall that we assume the circuit $C$ to consist of $d$ layers
each, and that each gate in layer $i$ depends on two outputs from from layer
$i-1$.

The high level strategy is to pack the inputs for each layer into blocks in
a way that allows to evaluate multiplication, addition, and subtraction gates in this layer ``in parallel''
on pairs of blocks. That is, the computation of the layer will consist of
disjoint parallel computations of the form $a\cdot b$, $a+b$, and $a-b$,   where $a$ and $b$
are blocks of $\ell$ binary values and the ring operation is performed
coordinate-wise. This will require blocks to be set up so that certain
inputs appear in several places.  Such a replication pattern will be
enforced using the procedure described above. Throughout the protocol, if a prover is caught cheating the protocol is aborted.

The protocol will proceed as follows:
\begin{enumerate}
  \item {\em Sharing inputs.} The clients arrange their inputs into blocks
  with a replication pattern that sets up the parallel evaluation for the
  first layer (namely, so that the first layer will be evaluated by applying the same arithmetic operation to blocks 1,2, to blocks 3,4, etc.). Each client then
  secret-shares its blocks, proving to the servers that the shares of each
  block agree with a polynomial of degree at most $\delta$ and that the secrets
  in the shared blocks satisfy the replication pattern determined by the
  first layer of $C$. (Such proofs are described in the previous section.)

If we want to enforce input values to be boolean (namely, either 0 or 1) this can be done a standard
  way by letting the servers
  securely reveal $1-a\cdot a$ for each block $a$ (which should evaluate to
  a block of $0$'s).

  \item {\em Evaluating $C$ on shared blocks.} The main part of the protocol
  is divided into $d$ phases, one for evaluating each layer of $C$. For
  $h=1,2,\ldots,d$ we repeat the following:
  \begin{itemize}
	\item  {\em Combining and and blinding.} At the beginning of the phase, the
	inputs to layer $h$ are arranged into blocks, so that the outputs of
	layer $h$ can be obtained by performing some arithmetic operation on each consecutive pair of
	blocks. Moreover, each block is secret-shared using a degree-$\delta$
	polynomial. Addition and subtraction on blocks can be handled non-interactively by simply letting each server locally add or subtract its two shares. In the following we address the more involved case of multiplication.
We would like to reveal the outputs of the layer to Alice,
	masked by random blinding blocks picked by Bob. For this, Bob will VSS
	random blocks, one for each block of output. The secret-sharing of these
	blocks is done using polynomials of degree $2\delta$.

(Again, verifying that the shares distributed by Bob are valid is done using the procedure described above.) For every pair of input blocks
	$a,b$ whose product is computed, each server $j$ locally computes the degree-2 function $c(j)=a(j)b(j)+r(j)$, where $a(j),b(j)$ are its
	shares of $a,b$ and $r(j)$ is its share of the corresponding blinding
	block $r$ distributed by Bob. For each pair of blocks combined in this
	way, the server sends his output (a single field element) to Alice. Note
	that the points $c(j)$ lie on a random degree-$2\delta$ polynomial $p_c$, and
	thus reveal no information about $a,b$. Moreover, the polynomial $p_c$
	can be viewed as some valid degree-$2\delta$ secret sharing of the block $c=ab+r$.

\item {\em Reducing degree and rearranging blocks for layer $h+1$}.
	Alice checks that the points $c(j)$ indeed lie on a polynomial $p_c$ of
	degree at most $2\delta$ (otherwise she aborts).
Then she recovers the blinded
	output block $c=ab+r$ by letting $c_j=p_c(1-j)$. Now Alice uses all
	blinded blocks $c$ obtained in this way to set up the (blinded) blocks
	for computing layer $h+1$.

For this, she sets up a new set of blocks that are obtained by applying a projection (namely, permuting and copying) to the blocks $c$ that corresponds to the structure of layer $h+1$. (In particular, the number
	of new blocks in which an entry in a block $c$ will appear is precisely
	the fan-out of the corresponding wire in $C$.) Let $c'$ denote the
	rearranged blinded blocks.

	Now Alice secret-shares each block $c'$ using a degree-$\delta$ polynomial $p_{c'}$. She needs to prove to the servers that the shares she
	distributed are of degree $\delta$ and that the entries of the shared blocks
	$c'$ satisfy the required relation with respect to the blocks $c$ that
	are already shared between the servers using degree-$2\delta$ polynomials.
	Such a proof can be efficiently carried out using the procedure
	described above. Note that pairs of polynomials $(p_c,p_{c'})$ such that
	$p_c$ is of degree at most $2\delta$, $p_{c'}$ is of degree at most $\delta$, and
	$p_c(i)=p_{c'}(j)$ form a linear space (for any fixed $i,j$), and hence
	the $2n$ evaluations of such polynomials on the points that correspond
	to the servers form a linear subspace of $F^{2n}$. Also, the
	corresponding code will have a large minimal distance because of the
	degree restriction, which ensures that the adversary cannot corrupt a
	valid codeword without being detected (or even corrected, in the setting of security without abort).

	\item {\em Unblinding.} To set up the input blocks for the evaluation of
	layer $h+1$, we need to cancel the effect of the blinding polynomials
	$p_r$ distributed by Bob. For this, Bob distributes random degree-$\delta$
	unblinding polynomials $p_{r'}$ that encode blocks $r'$ obtained by
	applying to the $r$ blocks the same projection defined by the structure
	of layer $h+1$ that was applied by Alice. Bob proves that the
	polynomials $p_{r'}$ are consistent with the $p_r$ similarly to the
	corresponding proof of Alice in the previous step. (In fact, both
	sharing the $p_{r'}$ and proving their correctness could be done in the
	first step.) Finally, each server obtains its share of an input block
	$a$ for layer $h+1$ by letting $a(j)=c'(j)-r'(j)$.
  \end{itemize}

  \item {\em Delivering outputs.} The outputs of $C$ are revealed to the
  clients by having the servers send their shares of each output block to
  the client who should receive it. The client checks that the $n$ values received for each block are consistent with a degree-$d$ polynomial (otherwise it aborts), and recovers the output of this block.

\end{enumerate}

\paragraph{Communication complexity.} By the choice of parameters, the communication overhead of encoding each block of field elements is constant. Accounting for narrow layers (whose size is smaller than one block) as well as wires between non-adjacent layers, we get an additive arithmetic communication overhead of $O(nd)$ (accounting for the worst-case scenario of one may need to maintain a block of values to be used in each subsequent layer). As noted above, this overhead can be reduced or even eliminated in most typical cases.
Finally, the cost of picking random field elements for the random linear combinations can be reduced via the use of (arithmetic) $\epsilon$-biased generators or directly improved via an alternative procedure described below.

\paragraph{Computational complexity.}
Using known FFT-based techniques for multipoint polynomial evaluation and interpolation, both the secret sharing and the reconstruction of a block of length $\ell$ with $n=O(\ell)$ servers can be done with arithmetic complexity of $O(\ell\log^2 \ell)$~\cite{vzGG99}.
Choosing evaluation points which are $n$-th roots of unity, this complexity can be reduced to $O(\ell\log \ell)$ (at the expense of sacrificing the black-box use of the field).
The computational bottleneck in the above protocol is the procedure for verifying that shared blocks satisfy the replication pattern corresponding to $C$. This can be improved by converting $C$ into an equivalent circuit $C'$ which reduces the overhead of this procedure. A more direct and efficient way for implementing the above procedure can be obtained by adapting an idea from~\cite{Groth08} to our setting. To test that a set of $M$ blocks $v_i$ satisfies a given replication pattern, pick a set of $M$ random blocks $r_i$ and test
that $\sum v_ir_i=\sum v_ir'_i$, where the blocks $r'_i$ are obtained by permuting the blocks in $r_i$ along the ``cycles'' defined by the replication pattern. (That is, for each set of positions in the blocks $v_i$ which should be tested to be equal, apply a cyclic shift to the values in the corresponding entries of the blocks $r_i$.) This sum of inner products can be computed by adding up all pointwise-products of $v_i$ and $r'_i$ together with a random block whose entries add up to 0.

\end{document}